\documentclass[11pt]{article}
\usepackage{fullpage}
\usepackage{macros}
\usepackage{authblk}
\usepackage{enumitem}
\usepackage{afterpage}
\usepackage{caption}
\usepackage{subcaption}
\usepackage{cancel}
\usepackage{amsmath}
\usepackage[numbers]{natbib}
\usepackage{hyperref}
\usepackage{thmtools}
\usepackage{thm-restate}
\usepackage[capitalise]{cleveref}

\hypersetup{colorlinks, citecolor=red, filecolor=blue, linkcolor=blue, urlcolor=blue}
\usepackage{enumitem}

\newcommand{\kcspip}{\ensuremath{k\operatorname{-CS-PIP}}\xspace}
\newcommand{\pip}{\ensuremath{\operatorname{PIP}}\xspace}
\newcommand{\pips}{\ensuremath{\operatorname{PIPs}}\xspace} 

\newcommand{\ufp}{\m{\operatorname{UFP-TREES}}\xspace}
\newcommand{\lca}{\m{\operatorname{LCA}}}

\newcommand{\bigc}{\ensuremath{\operatorname{big}}\xspace}
\newcommand{\medc}{\ensuremath{\operatorname{med}}\xspace}
\newcommand{\smallc}{\ensuremath{\operatorname{tiny}}\xspace}

\newcommand{\bigb}{\operatorname{BB}\xspace}
\newcommand{\medb}{\operatorname{MB}\xspace}
\newcommand{\smallb}{\operatorname{TB}\xspace}

\newcommand{\RF}{\cR_{F}}

\def \SOL {\ensuremath{\operatorname{SOL}}\xspace}

\newcommand{\func}{\textsc}

\newcommand{\wrt}{{\em w.r.t.~\xspace}}

\newcommand{\sksp}{\ensuremath{\operatorname{SKSP}}\xspace}
\newcommand{\ad}{E}    
\newcommand{\balp}{\bar{\alpha}}  
\newcommand{\bbeta}{\bar{\beta}}  
\def \HM {\ensuremath{\operatorname{HM}}\xspace}

\title{Algorithms to Approximate Column-Sparse Packing Problems\thanks{Extended abstract of this paper appeared in SODA-2018. Full version will appear in ACM Transactions of Algorithms (TALG). Research supported by NSF Awards CNS 1010789, CCF 1422569, CCF-1749864, a gift from Google and research awards from Adobe and Amazon.}
}
\author[1]{Brian Brubach\thanks{\textbf{Email: } \texttt{bbrubach@cs.umd.edu}}}
\author[1]{Karthik A. Sankararaman\thanks{\textbf{Email: }\texttt{karthikabinavs@gmail.com}}}
\author[1]{Aravind Srinivasan\thanks{\textbf{Email: } \texttt{srin@cs.umd.edu}}}
\author[1, 2]{Pan Xu\thanks{\textbf{Email: } \texttt{panxu@cs.umd.edu}}}
\affil[1]{Department of Computer Science, University of Maryland, College Park, USA}
\affil[2]{New Jersey Institute of Technology, Newark, NJ, USA}

\date{}

\begin{document}

\maketitle

\begin{abstract}
		Column-sparse packing problems arise in several contexts in both deterministic and stochastic discrete optimization. We present two unifying ideas, \emph{(non-uniform) attenuation} and \emph{multiple-chance algorithms}, to obtain improved approximation algorithms for some well-known families of such problems. As three main examples, we attain the integrality gap, up to
lower-order terms, for known LP relaxations for $k$-column-sparse packing integer programs (Bansal \etal, \emph{Theory of Computing}, 2012) and stochastic $k$-set packing (Bansal \etal, \emph{Algorithmica}, 2012), and go ``half the remaining distance" to optimal for a major integrality-gap 
conjecture of F\"uredi, Kahn and Seymour on hypergraph matching (\emph{Combinatorica}, 1993). \\

\end{abstract}

\thispagestyle{empty}
  \newpage

\setcounter{page}{1}
\pagenumbering{arabic}

\section{Introduction}

Column-sparse packing problems arise in numerous contexts (e.g., \cite{AGM15,BCNSX15,BGLMNR12,BKNS12,KS04,BS00,CH01,CL12,LLRS01,FKS93,Kl96}). 
We present two unifying ideas (attenuation and multiple-chances) to obtain improved approximation algorithms and/or (constructive) existence results for some well-known families of such problems. These two unifying ideas help better handle the \emph{contention resolution} \cite{CVZ14} that is implicit in such problems. As three main examples, we attain the integrality gap (up to
lower-order terms) for known LP relaxations for $k$-column-sparse packing integer
programs (\kcspip: Bansal \etal~\cite{BKNS12}) and stochastic
$k$-set packing (\sksp: Bansal \etal~\cite{BGLMNR12}), and go ``half
the remaining distance" to optimal for a major integrality-gap conjecture of 
F\"uredi, Kahn and Seymour on hypergraph matching \cite{FKS93}.

Letting $\bR_+$ denote the set of non-negative reals, a general \textit{Packing Integer Program} (\pip) takes the form,
\begin{equation}\label{eqn:pip}
\max \left\{ f(\vec{x}) \, | \, \vec{A} \cdot \vec{x} \leq \vec{b}, \vec{x} \in \{ 0, 1 \}^n \right\}, \text{ where } \vec{b} \in \bR_+^m, \text{ and } \vec{A} \in \bR_+^{m \times n};
\end{equation}	
here $\vec{A} \cdot \vec{x} \leq \vec{b}$ means, as usual, that $\vec{A} \cdot \vec{x} \leq \vec{b}$ coordinate-wise. Furthermore, $n$ is the number of variables/columns, $m$ is the number of constraints/rows, $\vec{A}$ is the matrix of sizes with the $j^{th}$ column representing the size vector $\SI_j \in \mathbb{R}_{+}^m$ of \textit{item} $j$, $\vec{b}$ is the \textit{capacity} vector, and $f$ is some non-decreasing function (often of the form $\vec{w} \cdot \vec{x}$, where $\vec{w}$ is a nonnegative vector of weights). The items' ``size vectors" $\SI_j$ can be deterministic or random. \pips generalize a large class of problems in combinatorial optimization. These range from optimally solvable problems such as classical matching to much harder problems like independent set which is NP-Hard to approximate to within a factor of $n^{1-\epsilon}$ \cite{Zu07}.

A $k$-column-sparse packing program ($k\operatorname{-CS-PP}$) refers to a special case of packing programs wherein each size vector $\SI_j$ (a column of $\vec{A}$) takes positive values only on a subset $\cC(j) \subseteq [m]$ of coordinates with $|\cC(j)| \le k$. The $k\operatorname{-CS-PP}$ family captures a broad class of packing programs that are well studied such as $k$-column-sparse packing integer programs (\kcspip), $k$-uniform hypergraph matching, stochastic matching, and stochastic $k$-set packing (\sksp). While we primarily focus on programs with linear objectives, some of these approaches can be extended to monotone submodular objectives as well from prior work (\eg \cite{BKNS12}, \cite{CVZ14}).

We show randomized-rounding techniques (including \textit{non-uniform attenuation}, \textit{multiple chances}) that, along with the ``nibble method'' \cite{DBLP:journals/jct/AjtaiKS80,DBLP:journals/ejc/Rodl85} in some cases, 
yield improved results for some important families of Packing Integer Programs (\pips). In the case of \kcspip and \sksp, we show \textit{asymptotically optimal} bounds matching the LP integrality gap (as a function of the column-sparsity $k$, which is our asymptotic parameter). For hypergraph matching, we make progress ``half the remaining way'' towards meeting a classic conjecture of F\"uredi \etal \cite{FKS93}. Additionally, we show a simple application of simulation-based attenuation to obtain improved ratios for the Unsplittable Flow Problem on trees (\ufp: Chekuri \etal~\cite{CMS07}) with unit demands and submodular objectives, a problem which admits a natural packing-LP relaxation.



%
%
%
%

\subsection{Preliminaries and Main Results}
\label{subsec:prelim}

%
%

The natural \LP relaxation is as follows (although
additional valid constraints are necessary for \kcspip~\cite{BKNS12}).
%
%
%
\begin{equation}\label{eqn:pip-lp}
\textstyle \max \{\vec{w} \cdot\vec{x} : \vec{A}\cdot\vec{x}\leq\vec{b}, \vec{x} \in [0, 1]^n\}.
\end{equation}	

Typically, a rounding algorithm takes as input an optimal solution $\vec{x}\in [0,1]^{n}$ to LP~\eqref{eqn:pip-lp} -- or one of
its relatives -- and outputs an integral $\vec{X}\in\{0,1\}^n$ which is feasible for \pip~\eqref{eqn:pip} such that the resultant \emph{approximation ratio}, $\frac{\bE[\vec{w} \cdot \vec{X}]}{\vec{w} \cdot \vec{x}}$, is maximized. 
Note that $\bE[\vec{w} \cdot \vec{X}]$ is the expected weight of the solution over all randomness in the algorithm and/or the problem itself. 
For a general weight vector $\vec{w}$, we often seek to maximize $\min_{j: x_j \neq 0} \frac{\bE[X_j]}{x_j}$, as the usual ``local" strategy of maximizing the approximation ratio. As notation, we will denote the support of $\vec{X}$ as the set of \emph{rounded} items. We say item $j$ \emph{participates} in constraint $i$ if and only if $A_{ij} \neq 0$. We say that a variable is \emph{safe} to be rounded to $1$ if doing so would not violate any constraint conditional on the variables already rounded; we call it \emph{unsafe} otherwise.

\xhdr{$k$-Column-Sparse Packing Integer Programs (\kcspip).} Suppose we have $n$ items and $m$ constraints. Each item $j \in [n]$ has a weight $w_j$ and a column $\vec{a}_j \in [0,1]^m$. Suppose we have a capacity vector $\vec{b}=\vec{1}$ (this is w.l.o.g., see \eg Bansal \etal~\cite{BKNS12}) and our goal is to select a subset of items such that the total weight is maximized while no constraint is violated. In addition, we assume each column $\vec{a}_j$ has at most $k$ non-zero entries. The important special case where for all $j$, $\vec{a}_j$ lies in $\{0,1\}^m$ -- and has at most $k$ non-zero entries -- is the classic \emph{$k$-set packing problem}; it is $\text{NP}$-hard to approximate within $o(k / \log k)$ \cite{HSS06}. This special case is generalized in two ways below: by allowing stochasticity in stochastic $k$-set packing, and by allowing the column-sparsity $k$ to vary across columns as in hypergraph matching). Observe that \kcspip can be cast as a special case of \pip shown in~\eqref{eqn:pip} with the $j^{\text{th}}$ column of $\vec{A}$ being $\vec{A}[j]=\vec{a}_j$. The resultant \LP relaxation is as follows (just as in Bansal \etal~\cite{BKNS12}, we will ultimately use a stronger form of this \LP relaxation which incorporates additional valid constraints; see (\ref{eq:addcons}) in Section~\ref{sec:kcspipalg}).
\begin{equation}\label{eqn:kcspip-lp}
\textstyle \max\{\vec{w} \cdot\vec{x} : \vec{A}\cdot\vec{x}\leq\vec{1}, \vec{x} \in [0, 1]^n\}~~~
\text{where }\vec{A}[j]=\vec{a}_j. 
\end{equation}	
For general \pips,  the best-known approximation bounds are shown in Srinivasan~\cite{Sr99}. The problem of \kcspip, in its full generality, was first considered by Pritchard~\cite{Pr09} and followed by several subsequent works such as Pritchard and Chakrabarty~\cite{PC11} and Bansal \etal~\cite{BKNS12}. Chekuri \etal~\cite{CVZ14} defined a contention resolution framework for submodular objectives and showed how the previous algorithms for \kcspip fit into such a framework (and hence, extending the \kcspip algorithms to non-negative submodular objectives by losing a constant factor in approximation)\footnote{In \cite{BKNS12}, the authors also show extensions to non-negative monotone submodular objectives.}. 

Our main result for this problem is described in Theorem~\ref{mainthm:kcspip}.
Bansal \etal~\cite{BKNS12} showed that the stronger \LP (which adds additional valid constraints to the natural \LP relaxation) has an integrality gap of at least $2k - 1$. We consider the same LP, and hence our result shown in \refthm{mainthm:kcspip} is asymptotically optimal \wrt this \LP. The previous best known results for this problem were a factor of $ek + o(k)$ due to Bansal \etal~\cite{BKNS12}, a factor of $O(k^2)$ independently due to Chekuri \etal~\cite{CEK09} \footnote{Note that this work is cited in \cite{BKNS12}.} and Pritchard and Chakrabarty~\cite{PC11}, and a factor of $O(2^k \cdot k^2)$ due to Pritchard~\cite{Pr09}. 

\begin{restatable}{theorem}{kcspipmain}
	\label{mainthm:kcspip}
	There exists a randomized rounding algorithm for \kcspip with approximation ratio at most $2k+ \Theta(k^{0.8} \poly \log (k)) = 2k + o(k)$ for linear objectives.
\end{restatable}

\begin{restatable}{corollary}{kcspipcor}
	\label{cor:kcspip}
	There exists a randomized rounding algorithm for \kcspip with approximation ratio at most $(2k+o(k))/\eta_f$ for non-negative submodular objectives, where $\eta_f$ is the approximation ratio for $\max\{F(\vec{x}) : \vec{x} \in \cP_{\cI} \cap \{0, 1\}^n\}$ (here, $F(\vec{x})$ is the multi-linear extension of the sub-modular function $f$ and $\cP_{\cI}$ is the \kcspip polytope); $\eta_f = 1 - 1/e$ and $\eta_f = 0.385$ in the cases of non-negative monotone and non-monotone submodular functions respectively\footnote{To keep consistent with prior literature, we state all approximation ratios for sub-modular maximization (\ie $\eta_f$) as a value less than $1$. This is in contrast to the approximation ratios defined in this paper where the values are always greater than $1$}.
\end{restatable}

\xhdr{Stochastic $k$-Set Packing (\sksp).} The Stochastic $k$-Set Packing problem was first introduced in Bansal \etal~\cite{BGLMNR12} as a way to generalize several stochastic-optimization problems such as Stochastic Matching\footnote{Here, we use the definition from the journal version~\cite{BGLMNR12}; the conference version of \cite{BGLMNR12} defines the problem slightly differently.}. The problem can be defined formally as follows. Suppose we have $n$ items and that each item $j$ has a random non-negative weight $W_j$ and a random $m$-dimensional size vector $\SI_j \in \{0,1\}^m$. The random variables $\{R_j := (W_j, \SI_j) : j \in [n]\}$ are mutually independent\footnote{Note that $W_j$ can be correlated with $\SI_j$.}. Each random vector $R_j \in \bR^+ \times \{0, 1\}^m$ is drawn from some probability distribution: our algorithm only needs to know the values of $u_{i,j} := \bE[\SI_{i, j}]$ for all $i,j$ -- where $\SI_{i, j}$ denotes the $i^{th}$ component of $\SI_j$ -- and $w_j := \bE[W_j]$.  Moreover, for each item $j$, there is a known subset $\cC(j) \subseteq [m]$ of at most $k$ coordinates such that $\SI_{i, j}$ can be nonzero only if $i \in \cC(j)$: all coordinates in $[m] \setminus \cC(j)$ will have value zero with probability $1$. We are given a capacity vector $\vec{b} \in \mathbb{Z}_{+}^{m}$. The algorithm proceeds in multiple steps. At each step, we consider any one item $j$ that has not been considered before, and which is safe with respect to the current remaining capacity, \ie adding item $j$ to the current set of already-added items will not cause any capacity constraint to be violated regardless of what random $\SI_j$ materializes.\footnote{This is called the \emph{safe-policy} assumption. This allows us to handle the correlations between $W_j$ and coordinates of $\SI_j$. A detailed discussion of this model can be found in~\cite{BGLMNR12}.} Upon choosing to probe $j$, the algorithm observes its size realization and weight, and has to irrevocably include $j$. The task is to sequentially probe some subset of the items such that the expected total weight of items added is maximized.

Let $\vec{w}$ denote $(w_1, \ldots,w_n)$ and $x_j$ denote the probability that $j$ is added in the \OPT solution.  Bansal \etal~\cite{BGLMNR12} introduced the following natural \LP to upper bound the optimal performance.
\begin{equation}\label{eqn:sksp-lp}
\textstyle \max\{ \vec{w} \cdot\vec{x} : \vec{A}\cdot\vec{x}\leq\vec{b}, \vec{x} \in [0, 1]^n\}~~~ \text{where } \vec{A}[i,j]= u_{i,j}.
\end{equation}

The previous best known bound for \sksp was $2k + o(k)$ due to Bansal \etal~\cite{BGLMNR12}. Our main contribution (Theorem~\ref{thm:sksp}) is to improve this bound to $k + o(k)$\footnote{In private communication from January 2018, \emph{Marek Adamczyk} has informed us that a result of $k + 1$ here can be obtained as a corollary of their recent work \kaedit{\cite{adamczyk2018random}}.}, a result that is again asymptotically optimal \wrt the natural \LP~\eqref{eqn:sksp-lp} considered (Theorem 1.3 from~\cite{FKS93}).

\begin{theorem}\label{thm:sksp}
	There exists a randomized rounding algorithm achieving an approximation ratio of $k+o(k)$ for the stochastic $k$-set packing problem, where the ``$o(k)$" is a vanishing term when $k \rightarrow \infty$.
\end{theorem}


\vspace{-3mm}
\xhdr{Hypergraph Matching.} Suppose we have a hypergraph $\cH=(\cV, \cE)$ with $|\cV|=m$ and $|\cE|=n$. (This is the opposite of the usual graph notation, but is convenient for us since the \LP here has $|\cV|$ constraints and $|\cE|$ variables.) Each edge $e \in \cE$ has a weight $w_e$. We need to find a subset of edges with maximum total weight such that every pairwise intersection is empty (\ie we obtain a hypergraph matching). Observe that the problem of finding a maximum weighted hypergraph matching can be cast as a special case of \pip. Let $\vec{w}=(w_e)$ and $\vec{e} \in \{0,1\}^m$ be the canonical (characteristic-vector) representation of $e$. 
Then the natural LP relaxation is as follows:
\begin{equation}\label{eqn:hm-lp}
\textstyle \max\{\vec{w} \cdot\vec{x} : \vec{A}\cdot\vec{x}\leq\vec{1}, \vec{x} \in [0, 1]^n\}~~~
\text{where }\vec{A}[j]=\vec{e}_j. 
\end{equation}	

Note that in these natural IP and LP formulations, the number of vertices in an edge $e$, $k_e=|e|$, can be viewed as the column-sparsity of the column associated with $e$. Thus, this again broadly falls into the class of column-sparse packing programs. For general hypergraphs, F\"uredi \etal~\cite{FKS93} presented the following well-known conjecture. 

\begin{conjecture}[F\"uredi \etal \cite{FKS93}]
	\label{conj:fks}
	For any hypergraph $\cH=(\cV,\cE)$ and a weight vector $\vec{w}=(w_e)$ over all edges, there exists a matching $\cM$ such that 
	\begin{equation}
	\label{eqn:fks}
	\textstyle \sum_{e \in \cM}\Paren{k_e-1+\frac{1}{k_e}}w_e \geq \OPT(\cH, \vec{w}).
	\end{equation}
	where $k_e$ denotes the number of vertices in hyperedge $e$ and $\OPT(\cH, \vec{w})$ denotes an optimal solution to the  
	LP relaxation (\ref{eqn:hm-lp}) of hypergraph matching. 
\end{conjecture}	

The function ``$k_e-1+\frac{1}{k_e}$" is best-possible in the sense that
certain hypergraph families achieve it \cite{FKS93}. 
We generalize Conjecture~\ref{conj:fks} slightly:

\begin{conjecture}[Generalization of Conjecture~\ref{conj:fks}]
	\label{conj:fks-ours}
	For any given hypergraph $\cH=(\cV,\cE)$ with notation as in
	Conjecture~\ref{conj:fks}, let $\vec{x} = (x_e: e \in \cE)$ denote a given optimal solution to the LP relaxation  (\ref{eqn:hm-lp}). Then: (i) there is a distribution $\cD$ on the matchings of $\cH$ such that for each edge $e$, the probability that it is present in a sample from $\cD$ is at least $\frac{x_e}{k_e - 1 + 1/k_e}$, and (ii) $\cD$ is efficiently samplable. 
\end{conjecture}

Part (i) of Conjecture \ref{conj:fks-ours} immediately implies Conjecture \ref{conj:fks} via the linearity of expectation. In fact, Part (i) of Conjecture \ref{conj:fks-ours} and Conjecture \ref{conj:fks} are equivalent which can be shown using a LP-duality argument for the standard LP for computing the fractional chromatic number (also see \cite{carrVempala} for such arguments in a general polytope).

F\"uredi \etal~\cite{FKS93} gave (non-constructive) proofs for Conjecture~\ref{conj:fks} for the three special cases where the hypergraph is either uniform, intersecting, or uniformly weighted. Chan and Lau~\cite{CL12} gave an algorithmic proof of Conjecture~\ref{conj:fks} for $k$-uniform hypergraphs, by combining the iterative rounding method and the fractional local ratio method. Using similar techniques, Parekh~\cite{Pa11} and Parekh and Pritchard~\cite{PP15} generalized this to $k$- uniform $b$-hypergraph matching \kaedit{and obtain the optimal integrality gap for this family. Moreover, they do so for the generalized Conjecture~\ref{conj:fks-ours}}. We go ``half the remaining distance" in resolving Conjecture \ref{conj:fks-ours} for \emph{all} hypergraphs, and also do so algorithmically: the work of 
Bansal \etal~\cite{BGLMNR12} gives ``$k_e + 1$" instead of the target
$k_e - 1 + 1/k_e$ in Conjecture \ref{conj:fks-ours}, and we improve this
to $k_e + O(k_e \cdot \exp(-k_e))$. 


\begin{theorem}\label{thm:hm}
	There exists an efficient \talgedit{randomized} algorithm to generate a matching $\cM$ for a hypergraph such that each edge $e$ is added in $\cM$ with probability at least $ \frac{x_e}{k_e+o(1)}$, where $\{x_e\}$ is an optimal solution to the standard \LP~\eqref{eqn:hm-lp} and where the $o(1)$ term is $O(k_e\exp(-k_e))$, a vanishing term when $k_e \rightarrow \infty$.
\end{theorem}


\vspace{-3mm}
\xhdr{\ufp with unit demands.}  In this problem, we are given a tree $T = ( \cV, \cE )$ with each edge $e$ having an integral capacity $u_e$. We are given $k$ distinct pairs of vertices $(s_1, t_1), (s_2, t_2), \ldots, (s_k, t_k)$ each having unit demand. Routing a demand pair $(s_i, t_i)$ exhausts one unit of capacity on all the edges in the path. With each demand pair $i$, there is an associated weight $w_i \ge 0$. The goal of the problem is to choose a subset of demand pairs to route such that no edge capacity is violated, while the total weight of the chosen subset is maximized. In the non-negative submodular version of this problem, we are given a non-negative submodular function $f$ over all subsets of demand pairs, and aim to choose a feasible subset that maximizes $f$. This problem was introduced by Chekuri \etal~\cite{CMS07} and the extension to submodular version was given by Chekuri \etal~\cite{CVZ14}. We show that by incorporating simple attenuation ideas, we can improve the analysis of the previous best algorithm for the Unsplittable Flow Problem in Trees (\ufp) with unit demands and non-negative submodular objectives.

Chekuri \etal~\cite{CVZ14} showed that they can obtain an approximation of $27/\eta_f$, where $\eta_f$ is the approximation ratio for maximizing a non-negative submodular function, via their contention-resolution scheme (henceforth abbreviated as \emph{CR schemes})\footnote{Please see Section~\ref{sec:kcs-sub} for formal definitions of CR schemes.}. We improve their $1/27$-balanced CR scheme to a $1/8.15$-balanced CR scheme via attenuation and hence achieve an approximation of $8.15/\eta_f$ for non-negative sub-modular objectives.
\begin{theorem}\label{thm:ufp}
	There exists a $8.15/\eta_f$-approximation algorithm to the \ufp with unit demands and non-negative submodular objectives.
\end{theorem}

\xhdr{Extension to submodular objectives.} 
Chekuri \etal~\cite{CVZ14} showed that given a rounding scheme for a \pip with linear objectives, we can extend it to non-negative submodular objectives by losing only a constant factor, if the rounding scheme has a certain structure (see \refthm{prelim:submodular}, due to~\cite{CVZ14}). Our improved algorithm for \kcspip and \ufp admits this structure and hence can be extended to non-negative sub-modular functions. See Section~\ref{appx:submodular} in the Appendix for the required background on submodular functions.

A simple but useful device that we will use often is as follows.

\xhdr{Simulation-based attenuation.} We use the term \emph{simulation} throughout this paper to refer to Monte Carlo simulation and the term \emph{simulation-based attenuation} to refer to the simulation and attenuation techniques as shown in \cite{AGM15} and \cite{BSSX17}\footnote{\kaedit{This is called ``dumping factor'' in \cite{AGM15}. See Appendix B in \cite{AGM15} for a formal treatment.}}. At a high level, suppose we have a randomized algorithm such that for some event $E$ (\eg the event that item $j$ is safe to be selected into the final set in \sksp) we have $\Pr[E] \geq c$, then we modify the algorithm as follows: (i) We first use simulation to estimate a value $\hat{E}$ that lies in the range $[\Pr[E], (1 + \epsilon) \Pr[E]]$ with probability at least $1 - \delta$. (ii) By ``ignoring" $E$ (i.e., \emph{attenuation}, in a problem-specific manner) with probability $\sim \talgedit{1-} c / \hat{E}$, we can ensure that the final effective value of $\Pr[E]$ is arbitrarily close to $c$, i.e., in the range $[c/(1 + \epsilon), c]$ with probability at least $1 - \delta$. This simple idea of attenuating the probability of an event to come down  approximately to a certain value $c$ is what we term simulation-based attenuation. The number of samples needed to obtain the estimate $\hat{E}$ is $\Theta(\frac{1}{c \epsilon^2} \cdot \log(\frac{1}{\delta}))$ via a standard Chernoff-bound argument. In our applications, we will take $\epsilon = 1/\mbox{poly}(N)$ where $N$ is the problem-size, and the error $\epsilon$ will only impact lower-order terms in our approximations.

\subsection{Our Techniques}
\label{sec:maintech}

In this section, we describe our main technical contributions of the paper and the ingredients leading up to them.

\xhdr{Achieving the integrality gap of the LP of \cite{BKNS12} for \kcspip.} Our first main contribution in this paper is to achieve the integrality gap of the strenghthened LP of \cite{BKNS12} for \kcspip, up to lower-order terms: we improve the $ek + o(k)$ of \cite{BKNS12} to $2k + o(k)$.  We achieve this by following the same overall structure as in \cite{BKNS12} and improve the alteration steps using randomization. We view the alteration step as a question on an appropriately constructed directed graph. In particular, a key ingredient in the alteration step answers the following question. ``Suppose we are given a directed graph $G$ such that the maximum out-degree is bounded by an asymptotic parameter $d$. Find a random independent set $\cI$ in the undirected version of this graph such that every vertex is added into $\cI$ with probability at least $1/(2d) - o(1/d)$''. It turns out that this question can be answered by looking at the more-general question of finding a good coloring of the undirected version of this graph. The key idea here is to ``slow down" the contention-resolution approach of \cite{BKNS12}, leading to Theorem~\ref{mainthm:kcspip}. However, motivated by works that obtain strong ``negative correlation" properties -- e.g., the papers \cite{CVZ10, PSV17} obtain \emph{negative cylindrical correlation}\footnote{This is sometimes simply called ``negative correlation".} and the even-stronger \emph{negative association} for rounding in matroid polytopes -- we ask next if one can achieve this for \kcspip. (It is well-known that even negative cylindrical correlation yields Chernoff-type bounds for sums of random variables \cite{PS97}; we use this in Section~\ref{sec:otherapp}.)  We make progress toward this in  Theorem~\ref{thm:negativeCorrelation}. 
	
	\xhdr{Achieving the integrality gap of the natural LP for \sksp via a ``multiple chances'' technique.} Our second contribution in is to develop an algorithm that achieves the integrality gap of $k + o(k)$ for \sksp, improving on the $2k$ of \cite{BGLMNR12}. To achieve this, we introduce the ``multiple-chances'' technique. We will now informally describe this technique, which is motivated by the powerful ``nibble'' idea from probabilistic combinatorics (see, e.g., Ajtai, Koml{\'{o}}s, and Szemer{\'{e}}di \cite{DBLP:journals/jct/AjtaiKS80} and R{\"{o}}dl \cite{DBLP:journals/ejc/Rodl85}). 
	
	The current-best ratios for many special cases of $k\operatorname{-CS-PP}$ are $\Theta(k)$; e.g., $ek + o(k)$ for \kcspip~\cite{BKNS12}, the optimal approximation ratio ({\it w.r.t.} the integrality gap) of $k-1+1/k$ for $k$-uniform hypergraph matching~\cite{CL12}, or the $2k$-approximation for \sksp~\cite{BGLMNR12}. Thus, many natural approaches involve sampling items with a probability proportional to $1/k$. Consider a $k\operatorname{-CS-PP}$ instance with budget $\vec{b}$. 
Suppose we have a randomized algorithm $\ALG$ which outputs a solution $\SOL$ wherein each item $j$ is added $\SOL$ with probability exactly \textit{equal} to $x_j/(ck)$ for some constant $c > 0$.\footnote{Note that given an algorithm which adds each item with probability \textit{at least} $x_j/(ck)$, we can use Monte Carlo simulation and attenuation techniques similar to \cite{CVZ14,AGM15,BSSX17} to get an algorithm which adds each item with probability essentially \textit{equal} to $x_j/(ck)$.} After running $\ALG$, the expected usage of each budget $i$ is $b_i/(ck)$; this follows directly from the budget constraint in the \LP. This implies that after running \ALG, we have only used a tiny fraction of the whole budget, in expectation. Thus, we may run \ALG again on the remaining items to further improve the value/weight of $\SOL$. Hence, an item that was previously not chosen, receives a ``second chance'' to be rounded up and included in the solution. The observation that only a tiny fraction of the budget is used can be made after running \ALG for a second time as well. Hence, in principle, we can run \ALG multiple times and we call the overarching approach a \emph{multiple chance algorithm}. The analysis becomes rather delicate as we run for a large number of iterations in this manner. 


\xhdr{FKS conjecture and the non-uniform attenuation approach.}
	Our third contribution is in making significant progress on the well-known Conjecture~\ref{conj:fks} due to F\"uredi, Kahn and Seymour. To achieve this, we introduce a technique of \emph{non-uniform} attenuation. A common framework for tackling $k\operatorname{-CS-PP}$ and related problems is random permutation followed by sampling via {\it uniform attenuation}: follow a random order $\pi$ on the items and add each item $j$ with probability $\alpha x_j$ whenever it is safe, where $x_j$ is an optimal solution to an appropriate \LP and $\alpha$ is the attenuation factor. Typically $\alpha$ is a parameter fixed in the analysis to get the best ratio (\eg see the \sksp algorithm in Bansal \etal~\cite{BGLMNR12}). This method is called {\it uniform attenuation}, since all items share the same attenuation factor $\alpha$. 
    
An alternative strategy used previously is that of {\it weighted random permutations} (see, e.g., Adamczyk \etal~\cite{AGM15} and Baveja \etal~\cite{BCNSX15}): instead of using a uniformly-random permutation, the algorithm ``weights'' the items and permutes them non-uniformly based on their $x_j$ values. We introduce a notion of \emph{non-uniform attenuation}, which approaches the worst-case scenario in a different manner. We still stay within the regime of uniform permutations but will attenuate items non-uniformly, based on their $x_j$ values; a careful choice of attenuation function is very helpful here, as suggested by the optimization problem (\ref{eqn:hm-main}). This is a key ingredient in our improvement.   

\subsection{Other Related Work}

	In this subsection, we list the related work not mentioned in previous sections and yet closely related to the problems we study. Note that packing programs, submodular maximization, and their applications to approximation algorithms have a vast literature. Our goal here is to list some papers in closely relevant areas and this is by no means an exhaustive list of references in each of these closely-aligned areas.
	
	For \kcspip, related problems have been studied in discrepancy theory. In such problems, we have a $k$-column-sparse \LP and we want to round the fractional solution such that the violation (both above and below) of any constraint  is minimized. This study started with the famous work of Beck and Fiala~\cite{BT81} and some of the previous work on \kcspip (\eg \cite{Pr09}) used techniques similar to Beck and Fiala. There has been a long line of work following Beck and Fiala, including \cite{Sr97, Ba10, BS13, LM15, Ro17, HSS14, BG16, LRR16, BDG16}. One special case of \kcspip is the $k$-set packing problem. Many works including \cite{HS89, AH98, CH01, Be00} studied this problem with \cite{Be00} giving the best approximation of $(k + 1)/2 + \epsilon$ for this problem. Closely related to \kcspip is the notion of column-restricted packing introduced by Kolliopoulos and Stein~\cite{KS04}. Many works have studied this version of packing programs, including \cite{CMS07, CEK09_2, BS00}. 
	
	Similar to Bansal \etal~\cite{BKNS12}, our algorithms also extend to submodular objective functions. In particular, we use tools and techniques from Calinescu \etal~\cite{CCPV11} and Chekuri \etal~\cite{CVZ14} for both \kcspip and the UFP problem on trees. Monotone sub-modular function maximization subject to $k$-sparse constraints has been studied in the context of $k$-partition matroids, $k$-knapsacks, and the intersection of $k$ partition matroids in  many works including \cite{FNW78, LMNS09, Wa12, KST09}. Beyond the monotone case, there are several algorithms for the non-negative sub-modular maximization problem including \cite{FNS11, CJV15, BFNS14, EN16, BF16}.
	
	Stochastic variants of \pips have also been previously studied. Baveja \etal~\cite{BCNSX15} considered the following stochastic setting of $k$-uniform hypergraph matching: the algorithm has to probe edge $e$ to check its existence; each edge $e$ is associated with a probability $0<p_e\le 1$ with which it will be present (independently of other edges) on being probed; the task is to sequentially probe edges such that the expected total weight of matching obtained is maximized.  The stochastic version of hypergraph matching can be viewed as a natural generalization of stochastic matching (\eg Bansal \etal~\cite{BGLMNR12}) to hypergraphs. The work of~\cite{BCNSX15} gave an $(k+\epsilon+o(1))$--approximation algorithm for any given $\epsilon>0$ asymptotically for large $k$. Other work on stochastic variants of \pips includes \cite{DGV05, DGV08, LY13, ASW14, GNS16, Ad15, GNS16_2}.

	 Later in this paper, we show yet another application of attenuation: \ufp with unit demands. This problem is a more specific version of column-restricted packing problems mentioned previously. The Unsplittable Flow Problem in general graphs and its various specializations on different kinds of graphs has been extensively studied. Some of these works include \cite{Kl96, Sr97_2, GVY97, Ko03, CKS06, BCES06, AR06, KS06, CCGK07, BFKS14, EGGKNP12}.

\subsection{Outline}

In Section~\ref{sec:kcspip}, we present a randomized rounding algorithm for \kcspip using randomized alteration techniques. We analyze this algorithm to prove Theorem~\ref{mainthm:kcspip} and show an extension to submodular objectives. 
In Section~\ref{sec:sksp}, we apply second-chance techniques to \sksp. After analyzing this algorithm, we show how it can be extended to multiple chances, yielding the improved result of Theorem~\ref{thm:sksp}. 
In Section~\ref{sec:hypermatch}, we present an algorithm for hypergraph matching and analyze it to prove Theorem~\ref{thm:hm}, making progress toward Conjecture~\ref{conj:fks-ours} (and by extension Conjecture~\ref{conj:fks}). 
In Section~\ref{sec:otherapp}, we show how attenuation can lead to an improved contention \-resolution scheme for \ufp, proving Theorem~\ref{thm:ufp}. 
We end with a brief conclusion and discussion of open problems in Section~\ref{sec:conclusion}. Appendix~\ref{sec:techlemmas} contains a few useful technical lemmas used in this paper while Appendix~\ref{appx:submodular} gives a self-contained background on submodular functions.

\section{$k$-Column-Sparse Packing}
	\label{sec:kcspip}
	We describe a rounding algorithm for \kcspip, which 
	achieves the asymptotically \emph{optimal} approximation ratio of $ (2k + o(k)) $ with respect to the strengthened \LP shown in Bansal \etal~\cite{BKNS12} (see (\ref{eq:addcons}) in Section~\ref{sec:kcspipalg}). Theorem~\ref{thm:negativeCorrelation} then develops a near-negative-correlation generalization of this result. 
	
	
	Recall that we have a $ k$-column-sparse matrix $ \vec{A} \in [0, 1]^{m \times n} $ and a fractional solution $ \vec{x} \in [0, 1]^n $ such that $ \vec{A} \cdot \vec{x} \leq \vec{1} $. Our goal is to obtain an integral solution $ \vec{X} \in \{0, 1\}^n $ (possibly random) such that $\vec{A} \cdot \vec{X} \leq \vec{1}$ and such that the expected value of the objective function $\vec{w}\cdot\vec{X}$ is ``large". (We will later extend this to the case where the objective function $f(\vec{X})$ is monotone submodular.) At a very high level, our algorithm performs steps similar to the contention-resolution scheme defined by Chekuri~\etal~\cite{CVZ14}; the main contribution is in the details\footnote{We would like to point out that the work of \cite{BKNS12} also performs similar steps and fits into the framework of \cite{CVZ14}.}. We first perform an independent-sampling step to obtain a random set $\cR$ of variables; we then conduct \emph{randomized} alterations to the set $\cR$ to obtain a set of rounded variables that are feasible for the original program with probability $1$. Note that the work of \cite{BKNS12} uses \textit{deterministic} alterations. Moving from deterministic alteration to careful randomized alteration combined with using a much-less aggressive uniform attenuation in the initial independent sampling, yields the optimal bound.

	\subsection{Algorithm}
	\label{sec:kcspipalg}
	
			Before describing the algorithm, we review some useful notations and concepts, some of which were introduced in~\cite{BKNS12}. 
			\talgedit{Let $\alpha > 0$ be a given parameter.\footnote{we later set it to be a constant (when discussing prior work) and to be $k^{0.4}$ in our algorithm.}} For a row $i$ of $\vec{A}$ and $\ell = \Theta(\log(k/\alpha))$, let $ \bigc(i) := \{j : a_{ij} > 1/2 \} $, $ \medc(i) := \{j : 1/\ell \leq a_{ij} \leq 1/2 \} $, and $ \smallc(i) := \{j : 0 < a_{ij} < 1/\ell\} $, which denote the set of \textit{big}, \textit{medium}, and \textit{tiny} items with respect to constraint $i$. For a given randomly sampled set $\cR$ and an item $j \in \cR$, we have three kinds of \textit{blocking events} for $j$. Blocking events occur when a set of items cannot all be rounded up without violating some constraint. In other words, these events may prevent $j$ from being rounded up. We partition the blocking events into the following three types:

			\begin{itemize}
				\item 
				$\bigb(j)$:  There exists some constraint $i$ with $a_{ij} > 0$ and an item $j' \neq j $ such that $ j' \in \bigc(i) \cap \cR$.
				\item
				$\medb(j)$: There exists some constraint $i$ with $\medc(i) \ni j$  such that $|\medc(i) \cap \cR\,| \ge 3$. 
			
				\item
				$\smallb(j)$: There exists some constraint $i$ with $ \smallc(i)  \ni j$ such that
				\[ \sum_{\mathclap{\substack{j' \neq j: \\  j' \in (\medc(i) \cup \smallc(i)) \cap \cR}}}a_{ij'} > 1- a_{ij} \text{ or } |\medc(i) \cap \cR\,| \ge 2. \]
			\end{itemize}

		Informally, we refer to the above three blocking events as the \emph{big, medium} and \emph{tiny} blocking events for $j$ with respect to $\cR$. 

			\xhdr{The main algorithm of Bansal \etal~\cite{BKNS12}.}  
			As briefly mentioned in Section~\ref{subsec:prelim}, Bansal \etal add certain valid constraints on big items to the natural \LP relaxation in~\eqref{eqn:kcspip-lp} as follows.
			\begin{equation}
				\label{eq:addcons}
				\max\{\vec{w} \cdot\vec{x} \quad \mbox{s.t.} \quad \vec{A}\cdot\vec{x}\leq\vec{1} \quad \mbox{and} \quad \forall i \in [m] \quad \sum_{j \in \bigc(i)}x_j \leq 1, ~\vec{x} \in [0, 1]^n\} \text{ where }\vec{A}[j]=\vec{a}_j. 
			\end{equation}

Algorithm~\ref{alg:bansal_kcspip}, $\func{BKNS}$, gives a formal description of the algorithm of Bansal \etal~\cite{BKNS12}, in which they set $\alpha=1$. 

			\IncMargin{1em}
			\begin{algorithm}[!h]
				\caption{$ \func{BKNS}(\alpha) $}
				\label{alg:bansal_kcspip}
				\DontPrintSemicolon
				\textbf{Sampling}: Sample each item $j$ independently with probability $(\alpha x_j)/k$ and let $\cR_0$ be the set of sampled items.\;\label{ln:bkns:sampling}
					\textbf{Discarding low-probability events}: Remove an item $j$ from $\cR_0$ if either a medium or tiny blocking event occurs for $j$ with respect to $\cR_0$. Let $\cR_1 \subseteq \cR_0$ be the set of items not removed.\;\label{ln:bkns:discard}
					\textbf{Deterministic alteration}: Remove an item $j$ from $\cR_1$ if a big blocking event occurs for $j$ with respect to $\cR_1$. \;\label{ln:bkns:alteration}
					Let $\RF \subseteq \cR_1$ be the set of items not removed; return $\RF$. 
			\end{algorithm}			
			\DecMargin{1em}
			
			\begin{theorem}[Bansal \etal~\cite{BKNS12}]
				\label{thm:bansal_kcspip}
By choosing $\alpha=1$, Algorithm~\ref{alg:bansal_kcspip} yields a randomized $ek +o(k)$-approximation for \kcspip.
			\end{theorem}

      \xhdr{Our algorithm for \kcspip via randomized alterations.}
			  \label{sec:kcspip:subsec:alg}
				Our pre-processing is similar to $\func{BKNS}$ with the crucial difference that $\alpha \gg 1$ (but not too large), i.e., we do not attenuate too aggressively; furthermore, our alteration step is quite different. Let $[n] = \{1, 2, \ldots, n\}$ denote the set of items. We first sample each item independently using an appropriate product distribution over the items (as mentioned above, we crucially use a different value for $\alpha$ than $\func{BKNS}$). Let $\cR_0$ denote the set of sampled items. We remove items $j$ from $\cR_0$ for which either a \emph{medium or tiny} blocking event occurs to obtain a set $\cR_1$. We next perform a randomized alteration, as opposed to a deterministic alteration such as in step \ref{ln:bkns:alteration} of $\func{BKNS}$. We then randomly and appropriately shrink $\cR_1$ to obtain the final set $\RF$.

			We now informally describe our randomized alteration step. We construct a directed graph $G = (\cR_1, E)$ from the constraints as follows. For every item $j \in \cR_1$, we create a vertex. We create a directed edge from item $j$ to item $j' \not= j$ in $G$ iff $j'$ causes a big blocking event for $j$ (\ie there exists a constraint $i$ where $j$ has a non-zero coefficient and $j'$ is in $\bigc(i)$). We claim that the expected \kaedit{out-}degree of every vertex in this graph constructed with $\cR_1$ is at most $\alpha$. If any vertex $j$ has \kaedit{out-}degree greater than $d := \alpha + \sqrt{\alpha \log(\alpha)}$, we will remove $j$ from $\cR_1$. Hence, we now have a directed graph with every vertex having \kaedit{out-}degree of at most $d$. We claim that we can color the undirected version of this directed graph with at most $2d + 1$ colors. We choose one of the colors $c$ in $[2d+1]$ uniformly at random and add all vertices of color $c$ into $\RF$. Algorithm~\ref{alg:kcspip} gives a formal description of our approach.

				\xhdr{Example.} Before moving to the analysis, we will show an example of how 
				the randomized alteration (\ie steps \ref{ln:alteration}(a-d) of Algorithm~\ref{alg:kcspip}) works. We will illustrate this on the integrality-gap example considered in \cite{BKNS12}. In this example, we have $n=2k-1$ items and $m=2k-1$ constraints. The weights of all items are $1$. For some $0 < \epsilon \ll 1/(nk)$, the matrix $\vec{A}$ is defined as follows. $\forall i, j \in [2k-1] $ we have,
				\[   		
							a_{ij} := \begin{cases}
								 1 & \quad \text{if $i=j$}\\
								 \epsilon &  \quad \text{if $j \in \{i+1, i+2, \ldots, i+k-1 \; (\bmod~n)\}$}\\
								 0 &  \quad \text{otherwise}  
							\end{cases}
				\]
				As noted in \cite{BKNS12}, setting $x_j = (1-k\epsilon)$ for all $j \in [n]$ is a feasible LP solution, while the optimal integral solution has value $1$. After running step \ref{ln:smapling} of the algorithm, each item $j$ is selected with probability $ (1-o(1)) \alpha/k$ independently. For simplicity, we will assume that there are no medium or tiny blocking events for every $j$ (these only contribute to the lower-order terms). Note that in expectation the total number of chosen items will be approximately $2\alpha$; with high probability, the total number of vertices in the graph will be $n_1 : =2\alpha + o(\alpha)$. Let $b_1, b_2, \ldots, b_{n_1}$ denote the set of items in this graph. The directed graph contains the edge $(b_i, b_j)$ for all distinct $i, j$; for simplicity, assume that the graph has no anomalous vertices. Since the undirected counterpart of this graph is a complete graph, every vertex will be assigned a unique color; thus the solution output will have exactly one vertex with probability $1 - o(1)$.

		\subsection{Analysis}
		We prove the following main theorem using Algorithm~\ref{alg:kcspip} with 
		\[ \alpha = k^{0.4}. \]
		\kcspipmain*
		
		We will divide the analysis into three parts. The parameters in the algorithm, namely $\ell, \alpha, d$ were chosen such that the $o(k)$ term in the final theorem is minimized. At a high-level the three parts prove the following.

		\begin{itemize}
				\item
						\textbf{Part 1} (Proved \footnote{\kaedit{Similar to the main theorem in \cite{graphcoloring}.}} in \cref{lemma:Color}). For directed graphs with maximum out-degree at most $d$, there exists a coloring $\chi$ and a corresponding algorithm such that the number of colors used, $|\chi|$, is at most $2d + 1$.
						
					\item
						\textbf{Part 2} (Proved in \cref{lemma:Anomalous}).					
								For any item $j \in \cR_1$, the event that the corresponding vertex in $G$ has an out-degree larger than $d$ occurs with probability at most $o(1)$. This implies that conditional on $j \in \cR_1$, the probability that $j$ is present in $G'$ is $1- o(1)$.
								
					\item
						\textbf{Part 3} (Proved in \cref{lemma:LowProbability}). For each item $j \in \cR_0$, either a medium or a tiny blocking event occurs with probability at most $o(1)$ (again, for our choice $\alpha= k^{0.4}$). This implies that for each $j \in \cR_0$, it will be added to $\cR_1$ with probability $1-o(1)$. 

				\end{itemize}

\hspace{-0.7cm}
\begin{minipage}[c]{0.97\linewidth}
	\IncMargin{1em}
	\begin{algorithm}[H]
		\caption{The Algorithm for \kcspip}
		\label{alg:kcspip}
		\DontPrintSemicolon
 
 	\textbf{Sampling}: Sample each item $j$ independently with probability $\alpha x_j/k$ (where $\alpha = k^{0.4}$) and let $\cR_0$ be the set of sampled items.\;\label{ln:smapling}
 	\textbf{Discarding low-probability events}: Remove an item $j$ from $\cR_0$ if either a medium or tiny blocking event occurs for $j$ with respect to $\cR_0$. Let $\cR_1$ be the set of items not removed.\;\label{ln:discard1}
 	\textbf{Randomized alteration}: 
 		{
 		\begin{addmargin}[0pt]{0.05\linewidth}
 		\begin{enumerate}[label=(\alph*)]
 			\item 
 					\label{ln:G'}\textbf{Create a directed graph}: For every item in $\cR_1$, create a vertex in graph $G$. Add a directed edge from item $j$ to item $j' \neq j$ if there exists a constraint $i$ such that $a_{ij} > 0$ and $a_{ij'} \talgedit{>} 1/2$.
 				\item
 					\label{ln:discard2}\textbf{Removing anomalous vertices}: For every vertex $v$ in $G$, if the out-degree of $v$ is greater than $d:=\alpha + \sqrt{\alpha \log(\alpha)}$, call $v$ \emph{anomalous}. Remove all anomalous vertices from $G$ to obtain $G'$ and let the items corresponding to the remaining vertices in $G'$ be $\cR_2$.
 				\item 
 					\label{ln:color}\textbf{Coloring $G'$}: Assign a coloring $\chi$ to the vertices of $G'$ using $2d+1$ colors as described in the text such that for any edge $e$ (ignoring the direction), both end points of $e$ receive different colors.
 				\item
 					\label{ln:IS}\textbf{Choosing an independent set}: Choose a number $c \in [2d + 1]$ uniformly at random. Add all vertices $v$ from $G'$ into $\RF$ such that $\chi(v) = c$. 
 		\end{enumerate}
 	\end{addmargin}
 	}\label{ln:alteration}
 	Return $\RF$. \;
\end{algorithm}	
\DecMargin{1em} 
\end{minipage}


We assume the following lemmas which are proven later in this section.

\begin{lemma}
	\label{lemma:Color}
		Given a directed graph $G = (V, E)$ with maximum out-degree at most $d$, there is a polynomial-time algorithm that finds a coloring $\chi$ of $G$'s undirected version such that $|\chi|$, the number of colors used by $\chi$, is at most $2d+1$.
\end{lemma}

\begin{lemma}
	\label{lemma:Anomalous}
		\kaedit{For any item $j$ we have $\Pr[j \in \cR_2 \given{j \in \cR_1}] = 1-o(1)$.}
\end{lemma}		

\begin{lemma}
	\label{lemma:LowProbability}
	\kaedit{For each item $j$ we have the following.
	\[
			\Pr[\text{Medium or tiny blocking event occurs for $j$} \given{j \in \cR_0}] \leq O(\alpha^2 k^{-1} \log^2(k/\alpha)) = o(1).
	\]
	}
\end{lemma}

We can now prove the main theorem, Theorem \ref{mainthm:kcspip}. 

\begin{proof}
First we show that $\RF$ is feasible for our original IP. We have the following observations about Algorithm~\ref{alg:kcspip}: (i) from step \ref{ln:discard1}, ``Discarding low-probability events'', we have that no item in $\RF$ can be  blocked by either medium or tiny blocking events; (ii) from the ``Randomized alteration'' steps in step \ref{ln:alteration}, we have that no item in $\RF$ has any neighbor in $G'$ that is also included in $\RF$. This implies that no item in $\RF$ can be blocked by any big blocking events. Putting together the two observations implies that $\RF$ is a feasible solution to our IP.

We now show that the probability of $j$ being in $\RF$ can be calculated as follows.
\begin{align*}
\Pr[j \in \RF] &= \Pr[j \in \cR_0] \cdot \Pr[j \in \cR_1 \given{j \in \cR_0}] \cdot
 \Pr[j \in \cR_2 \given{j \in \cR_1}] \cdot \Pr[j \in \RF \given{j \in \cR_2}],\\
&\geq \frac{\alpha x_j}{k} \cdot (1-o(1)) \cdot (1-o(1)) \cdot \frac{1}{2\alpha+2\sqrt{\alpha \log \alpha} + 1},\\
& = \frac{x_j}{2k(1+o(1))}. 
\end{align*}
The first inequality is due to the following. From the sampling step we have that $\Pr[j \in \cR_0] = \alpha x_j/k$. From Lemma~\ref{lemma:LowProbability} we have that $\Pr[j \in \cR_1 \given{j \in \cR_0}] = 1 - o(1)$. Lemma~\ref{lemma:Anomalous} implies that $\Pr[j \in \cR_2 \given{j \in \cR_1}] = 1 - o(1)$. Finally from Lemma~\ref{lemma:Color} we have that the total number of colors needed for items in $\cR_2$ is at most $2d+1$, and hence the probability of picking $j$'s color class is $1/(2d+1)$. Thus, $\Pr[j \in \RF \given{j \in \cR_2}] = 1/(2\alpha+2\sqrt{\alpha \log \alpha} + 1)$ (recall that $d:= \alpha + \sqrt{\alpha \log \alpha}$).

	\kaedit{We obtain Theorem 1 by using linearity of expectation. In other words, $\mathbb{E}[\vec{w} \cdot \vec{X}] \geq \left( \vec{w} \cdot \vec{x} \right) \tfrac{1}{2k(1+o(1))}$. Note that $\vec{w} \cdot \vec{x}$ is the optimal value to the LP~\eqref{eq:addcons}.}
\end{proof}

\xhdr{Proof of \cref{lemma:Color}.} We will prove this Lemma by giving a coloring algorithm that uses at most $2d + 1$ colors and prove its correctness. Recall that we have a directed graph such that the maximum out-degree $\Delta \leq d$. The algorithm is a simple greedy algorithm, which first picks the vertex with minimum \emph{total} degree (\ie sum of in-degree plus out-degree). It then removes this vertex from the graph and recursively colors the sub-problem. Finally, it assigns a color to this picked vertex not assigned to any of its neighbors. Algorithm~\ref{alg:directedColoring} describes the algorithm formally.

\let\oldnl\nl
\newcommand\nonl{%
	\renewcommand{\nl}{\let\nl\oldnl}}

\IncMargin{2em}
\begin{algorithm}[!h]
	\caption{Greedy algorithm to color bounded degree directed graph}
	\label{alg:directedColoring}
	\DontPrintSemicolon
	\SetKwFunction{greedyColor}{Color-Directed-Graph}
	\Indm\nonl\greedyColor{$G, V, d, \chi$}\\
	\Indp
	\If{$V = \phi$}{\KwRet{$\chi$}}
	\Else{
		Let $v_{\min}$ denote the vertex with minimum total degree.\;
		$\chi$ = \greedyColor{$G, V \setminus \{v_{\min}\}, d, \chi$}.\;
		Pick the smallest color $c \in [2d + 1]$ that is not used to color any of the neighbors of $v_{\min}$. Let $\chi(v_{\min}) = c$. \; \label{ln:dag:color}
		\KwRet{$\chi$}\;
	}
\end{algorithm}			
\DecMargin{2em}

We will now prove the correctness of the above algorithm. In particular, we need to show that in every recursive call of the function, there is always a color $c \in [2d + 1]$ such that the assignment in line \ref{ln:dag:color} of the algorithm is feasible. We prove this via induction on number of vertices in the graph $G$. 

\xhdrem{Base Case}: The base case is the first iteration is when the number of vertices is $1$. In this case, the statement is trivially true since $v_{\min}$ has no neighbors.

\xhdrem{Inductive Case}: We have that $\Delta \leq d$ for every recursive call. Hence, the sum of \emph{total} degree of all vertices in the graph is $2nd$ (Each edge contributes $2$ towards the total degree and there are $nd$ edges). Hence, the average total degree is $2d$. This implies that the minimum total degree in the graph is at most $2d$. Hence, the vertex $v_{\min}$ has a total degree of at most $2d$. From inductive hypothesis we have that $V \setminus \{v_{\min}\}$ can be colored with at most $2d + 1$ colors. Hence, there exists a color $c \in [2d + 1]$, such that $\chi(v_{\min}) = c$ is a valid coloring (since $v_{\min}$ has at most $2d$ neighbors). 

\xhdr{Proof of \cref{lemma:Anomalous}.} Consider an item $j \in \cR_1$. We want to show that the $\Pr[\delta_j > \alpha + \sqrt{\alpha \log \alpha}] \leq o(1)$, where $\delta_j$ represents the out-degree of $j$ in the directed graph $G$. Recall that from the construction of graph $G$, we have a directed edge from item $j$ to item $j'$ if and only if there is a constraint $i$ where $a_{ij'} > 1/2$ and $a_{ij} > 0$. For sake of simplicity, let $\mathcal{S} := \{1, 2, \ldots, N_j\}$ denote the set of 
	\kaedit{items that satisfy the following. For every $j' \in \mathcal{S}$, there is some constraint $i$ such that $a_{ij} > 0$ and $a_{ij'} > 1/2$.} Let $\{X_1, X_2, \ldots, X_{N_j}\}$ denote the corresponding indicator random variable for them being included in $\cR_1$. Hence, for every $i \in [N_j]$ we have $\mathbb{E}[X_i] = \alpha (1-o(1)) x_i/k$. From the strengthened constraints in \LP~\eqref{eq:addcons} \kaedit{and the $k$-column sparsity assumption}, we have that $\mathbb{E}[\delta_j] = \mathbb{E}[X_1 + X_2 + \ldots + X_{N_j}] \leq \alpha(1-o(1))$. Hence we have,
\begin{equation*}
	\Pr[\delta_j > \alpha + \sqrt{\alpha \log \alpha}] \leq \Pr[\delta_j > \mathbb{E}[\delta_j] + \sqrt{\alpha \log \alpha}]  \leq e^{-\Omega(\log \alpha)} = o(1).
\end{equation*}
The last inequality is from the Chernoff bounds, while the last equality is true for $\alpha = k^{0.4}$.

\xhdr{Proof of \cref{lemma:LowProbability}.} Consider the medium blocking event $\medb(j)$. Let $i$ be a constraint that causes $\medb(j)$ and let $j_1, j_2, \kaedit{\ldots, j_h} \neq j$ be the other variables \kaedit{in constraint $i$} such that $j_1, j_2, \kaedit{\ldots, j_h} \in \medc(i)$. Denote $X_{j}$, $X_{j_1}, X_{j_2}, \kaedit{\ldots,X_{ j_h}}$ to be the indicators that $j \in \cR_0, j_1, j_2, \kaedit{\ldots, j_h} \in \cR_0$ respectively.
	
	We know that $a_{ij}x_{j} + a_{ij_1}x_{j_1} + a_{ij_2}x_{j_2} \kaedit{+ \ldots + a_{ij_h}x_{j_h}} \leq 1$ and since $j, j_1, j_2, \kaedit{\ldots, j_h} \in \medc(i)$, we have 
	$x_{j} + x_{j_1} + x_{j_2} + \kaedit{\ldots + x_{j_h}} \leq \ell$ for some fixed value $\ell$. Scenario $\medb(j)$ is ``bad'' if $X_{j_1} + X_{j_2} + \kaedit{\ldots + X_{j_h}} \geq 2$. Note that $\mathbb{E}[X_j + X_{j_1} + X_{j_2} + \kaedit\ldots + {X_{j_h}}] \leq \frac{\alpha \ell}{k}$.
	
	Using the the Chernoff bounds in the form denoted in Theorem \ref{appx:chernoff} of Appendix, we have,
	\begin{equation*}
	\Pr[X_{j_1} + X_{j_2} + \kaedit{\ldots + X_{j_h}} \geq 2 \bigm| X_{j} = 1] = \Pr[X_{j_1} + X_{j_2} + \kaedit{\ldots + X_{j_h}} \geq 2] \leq O\Paren{\frac{\alpha^2 \ell^2}{k^2}}.
	\end{equation*}
	Note that the first equality is due to the fact that these variables are independent. Using a union bound over the $k$ constraints $j$ appears in, the total probability of the ``bad'' event is at most $O(\frac{\alpha^2 \ell^2}{k})$. And since $\alpha = k^{0.4}$ and $\ell = \Theta(\log(k/\alpha))$, this value is $o(1)$.
	
	For scenario $\smallb(j)$ we will do the following. If $j$ is tiny, \kaedit{a sufficient condition for the} ``bad'' event for constraint $i$ is \kaedit{if} one of the following occurs.
\kaedit{
	\begin{enumerate}
		\item \textbf{Blocked by other tiny items.} The sum of all coefficients of items $j'$ in $\cR_0$ where $a_{i, j'} \leq 1/\ell$ is greater than $1-1/\ell$.
		\item \textbf{Blocked by one medium item and other tiny items.} There exists an item $j'$ in $\cR_0$ such that $1/\ell < a_{i, j'} \leq 1/2$ and sum of coefficients of all items $j''$ such that $a_{i, j''} \leq 1/\ell$ is greater than $1/2-1/\ell$.
		\item \textbf{Blocked by two medium items.} There exists at least two items $j'$ and $j''$ in $\cR_0$ such that $1/\ell < a_{i, j'}, a_{i, j''} \leq 1/2$.
	\end{enumerate}
}	
	Note in case there there exists three or more medium items blocking a tiny item, that is handled in the previous case. Hence the only possible way a tiny item could be blocked after this is by a big item. We will now show that each of these cases occurs with probability $o(1)$.
	
	\emph{Case (1).} Consider a constraint $i$. Let $\cH_i$ denote the set of tiny items appearing in this constraint except item $j$. Note that $\mathbb{E}[\sum_{h \in \cH_i}A_{ih}X_{ih}] \leq \alpha/k$. Moreover we have,
		\[ \textstyle \Pr[\sum_{h \in \cH_i} A_{ih}X_{ih} > 1 - A_{ij}] \leq \Pr[\sum_{h \in \cH_i} A_{ih}X_{ih} > 1 - 1/\ell] = \Pr[\sum_{h \in \cH_i} \ell A_{ih}X_{ih} > \ell-1]. 
		\]
		\kaedit{Note that for every $h \in \cH_i$ we have $\ell A_{ih}X_{ih} \in [0, 1]$. Define $Y_h := \ell A_{ih}X_{ih}$. Thus, we have $\mathbb{E}\left[ \sum_{h \in \cH_i} Y_h \right] \leq \tfrac{\alpha \ell}{k}$. We want to upper-bound the quantity,
		\[
				\textstyle \Pr\left[ \sum_{h \in \cH_i} Y_h > \ell -1 \right].
		\]
		Define $\delta = \tfrac{\ell-1}{\ell} \cdot \tfrac{k}{2 \alpha}$. Note that $(1+ \delta)\tfrac{\alpha \ell}{k} \leq \ell - 1$. Thus, 
		\[
				\textstyle \Pr\left[ \sum_{h \in \cH_i} Y_h > \ell -1 \right] \leq \Pr\left[ \sum_{h \in \cH_i} Y_h > (1+ \delta)\tfrac{\alpha \ell}{k} \right].
		\]
		From Chernoff-Hoeffding bounds (Theorem~\ref{appx:chernoffmult}) we have that,
		\[
			\textstyle \Pr\left[ \sum_{h \in \cH_i} Y_h > (1+ \delta)\tfrac{\alpha \ell}{k} \right] \leq \exp\left[ - \tfrac{\delta^2}{2 + \delta} \cdot \tfrac{\alpha \ell}{k} \right].
		\]
		
		In what follows, we assume that $k \geq 2$, $\ell \geq 3$ and $\alpha \leq k$. Hence, $1 \geq \tfrac{\ell-1}{\ell} \geq \tfrac{2}{3} \geq \tfrac{1}{2}$.
		
		This implies that we have,
		\[
				\textstyle \tfrac{\delta^2}{2 + \delta} \cdot \tfrac{\alpha \ell}{k} \geq \left( \tfrac{\ell-1}{\ell} \right)^2 \cdot \tfrac{k^2}{4 \alpha^2} \cdot \tfrac{\alpha \ell}{k} \cdot \tfrac{1}{2 + \tfrac{\ell-1}{\ell} \cdot \tfrac{k}{2 \alpha} } \geq \tfrac{\ell}{40}.
		\]
		
		Combining the above arguments we obtain,
		\[
				\textstyle \Pr[\sum_{h \in \cH_i} \ell A_{ih}X_{ih} > \ell-1] \leq \exp\left[ - \tfrac{\ell}{40} \right].
		\]
		}
	 \\
	
	\emph{Case (2).} Consider a constraint $i$. Let $\cH_i$ denote the set of tiny items appearing in this constraint except the item $j$. Let $j'$ denote the medium item present in this constraint. Note that $\mathbb{E}[\sum_{h \in \cH_i}A_{ih}X_{ih}] \leq \alpha/k$. Moreover we have,
		\[\textstyle \Pr[\sum_{h \in \cH_i} A_{ih}X_{ih} > 1-A_{ij}-A_{ij'}] \leq \Pr[\sum_{h \in \cH_i} A_{ih}X_{ih} > 1/2-1/\ell] = \Pr[\sum_{h \in \cH_i} \ell A_{ih}X_{ih} > \ell/2-1].
		\]
		\kaedit{
			As in case (1) we invoke the Chernoff-Hoeffding bounds to find an upper-bound. Define $\delta := \tfrac{\ell/2 -1}{\ell} \cdot \tfrac{k}{2 \alpha}$. This implies that $(1+\delta) \tfrac{\ell \alpha}{k} \leq \ell/2 -1$. Therefore we obtain,
			\[
					\textstyle \Pr[\sum_{h \in \cH_i} \ell A_{ih}X_{ih} > \ell/2-1] \leq \exp\left[ - \tfrac{\delta^2}{2 + \delta} \cdot \tfrac{\alpha \ell}{k} \right].
			\]
			
			Moreover, the fact that $k \geq 2$ and $\ell \geq 3$ implies that $\tfrac{1}{2} \geq \tfrac{\ell/2-1}{\ell} \geq \tfrac{1}{6}$. Thus,
			\[
					\textstyle \tfrac{\delta^2}{2 + \delta} \cdot \tfrac{\alpha \ell}{k} \geq  \left( \tfrac{\ell/2-1}{\ell} \right)^2 \cdot \tfrac{k^2}{4 \alpha^2} \cdot \tfrac{\alpha \ell}{k} \cdot \tfrac{1}{2 + \tfrac{\ell/2-1}{\ell} \cdot \tfrac{k}{2 \alpha} } \geq \tfrac{\ell}{36}.
			\]
			
			Combining the above arguments we obtain that,
			\[
					\textstyle \Pr[\sum_{h \in \cH_i} \ell A_{ih}X_{ih} > \ell/2-1] \leq \exp\left[ - \tfrac{\ell}{36} \right].
			\] 
		}
	\emph{Case (3).} Consider a constraint $i$. Let $\cG_i$ denote the set of medium items in this constraint. The condition we want is $\sum_{g \in \cG_i} X_{ig} \geq 2$. Note that $\mathbb{E}[\sum_{g \in \cG_i} X_{ig}]\leq \alpha \ell /k$. Using the Chernoff bounds of the form for the medium case, we have that $\Pr[\sum_{g \in \cG_i} X_{ig} \geq 2] \leq O\Paren{\frac{\alpha^2 \ell^2}{k^2}}$.
	
	Note that taking a union bound over the $k$ constraints, setting \kaedit{$\ell = 80 \log(k/\alpha)$} and $\alpha = k^{0.4}$, we have that the probability of the tiny blocking event (Case (3)) occurring is $O(\frac{\alpha^2 \ell^2}{k})$ = $o(1)$. \kaedit{Likewise, the upper-bounds in Case (1) and Case (2) evaluates to $k^{-0.2}$ and $k^{-1/3}$ respectively. Each of these quantities are $o(1)$.}

\xhdr{\emph{Near}-negative correlation.} The natural approach to proving Lemma~\ref{lemma:Color} can introduce substantial positive correlation among the items included in $\cR_{F}$. However, by slightly modifying the coloring algorithm, we obtain \emph{near}-negative correlation in the upper direction among the items. \kaedit{In particular, we first color using the modified coloring scheme in Algorithm~\ref{alg:directedColoringNegCor}. Then, we choose $c$ uniformly at random from the set $[2d + d^{1-\epsilon}]$ and include all items with the color $c$ into the set $\cR_{F}$.}
	
	\IncMargin{2em}
	\begin{algorithm}[!h]
		\caption{Greedy algorithm to color bounded out-degree directed graphs using $2d + d^{1-\epsilon}$ colors, and with near-negative correlation}
		\label{alg:directedColoringNegCor}
		\DontPrintSemicolon
		\SetKwFunction{greedyColor}{Color-Dir-Graph-Neg-Corr}
		\Indm\nonl\greedyColor{$G, V, d, \epsilon, \chi$}\\
		\Indp
		\If{$V = \phi$}{\KwRet{$\chi$}}
		\Else{
			Let $v_{\min}$ denote the vertex with minimum total degree.\;
			$\chi$ = \greedyColor{$G, V \setminus \{v_{\min}\}, d, \epsilon, \chi$}.\;
			Among the \emph{smallest} $d^{1-\epsilon}$ colors in the set $[2d + d^{1-\epsilon}]$ that have not been used thus far to color any of the neighbors of $v_{\min}$, choose a color $c_r$ \emph{uniformly} at random. Let $\chi(v_{\min}) = c_r$. \;
			\KwRet{$\chi$}\;
		}
	\end{algorithm}			
	\DecMargin{2em}

\kaedit{Negative correlation among the events that various items are rounded up is a desirable property, since it reduces the overall variance of the objective function. Recall that in Theorem~\ref{mainthm:kcspip} the ratio holds in expectation. However, the significant positive correlation implies that the variance can be large. Thus, by achieving near negative correlation we can reduce the variance. In particular, by using $t=2$ in Theorem~\ref{thm:negativeCorrelation} we immediately obtain a bound on the variance. In fact, using prior work it also implies strong concentration bounds (see Remark~\ref{rem:concNeg} below).}

\begin{theorem}
	\label{thm:negativeCorrelation}
	Given any constant $\epsilon \in (0,1)$, there is an efficient randomized algorithm for rounding a fractional solution within the \kcspip polytope, such that 
	\begin{enumerate}
		\item For all items $j \in [n]$, $\Pr[j \in \cR_{F}] \geq \frac{x_j}{2k(1+o(1))}.$
		\item For any $t \in [n]$ and any $t$-sized subset $\{v_1, v_2, \ldots, v_t \}$ of items in $[n]$, we have (with $d =  \alpha + \sqrt{\alpha \log \alpha}$ \talgedit{and $\alpha=k^{0.4}$} as above)
			\begin{equation*}
			\Pr[v_1 \in \cR_{F} \wedge v_2 \in \cR_{F} \wedge \ldots \wedge v_t \in \cR_{F}] \leq (2 d^{\epsilon})^{t-1} \cdot \prod_{j = 1}^t \frac{x_{v_j}}{2k}. 
			\end{equation*}
	\end{enumerate}
\end{theorem}	
\begin{proof}
	\kaedit{Consider the coloring procedure in Algorithm~\ref{alg:directedColoringNegCor}. We will show that this induces near-negative correlation. The other steps in the proof remain the same and will directly imply the theorem.} Note that after coloring using Algorithm~\ref{alg:directedColoringNegCor}, we obtain the rounded items by first choosing one of the colors $c$ in the set $[2d + d^{1-\epsilon}]$ uniformly at random and then add all the vertices which received a color $c$ into the set $\cR_{F}$. Since $\alpha^{1-\epsilon} \leq o(\alpha)$, a similar analysis as before follows to give part (1) of Theorem~\ref{thm:negativeCorrelation}. We will now prove part (2). Fix
		an arbitrary $t \in [n]$ and any $t$-sized subset $U := \{v_1, v_2, \ldots, v_t \}$ of items in $[n]$. A necessary condition for these items to be present in $G'$ is that they were all chosen into $\cR_0$, which happens with probability 
        \begin{equation}
        \label{eqn:all-in-G'}
        \textstyle \prod_{j=1}^t \frac{x_{v_j} \alpha}{k};
        \end{equation}
        suppose that this indeed happens (all the remaining probability calculations are conditional on this). Note that Algorithm~\ref        {alg:directedColoringNegCor} first removes the vertex $v_{\min}$ and then recurses; i.e., it removes the vertices one-by-one, starting with $v_{\min}$. Let $\sigma$ be the \emph{reverse} of this order, and suppose that the order of vertices in $U$ according to $\sigma$ is, without loss of generality, $\{v_1, v_2, \ldots, v_t \}$. Recall again that our probability calculations now are conditional on all items in $U$ being present in $G'$; denote \kaedit{this event by $U_{G'}$ and $\mathbb{I}[U_{G'}]$ to be the corresponding indicator}. Note that $\kaedit{\Pr[v_1 \in \cR_{F}\given{\mathbb{I}[U_{G'}]}]} = \frac{1}{2d + d^{1-\epsilon}} \leq \frac{1}{2\alpha}$. \kaedit{This is because we choose one of the colors in $[2d + d^{1-\epsilon}]$ colors at random for the first vertex.} Next, a moment's reflection shows that for any $j$ with $2 \leq j \leq t$, 
        \begin{equation*}
		\kaedit{\Pr[v_j \in \cR_{F} \given{v_1 \in \cR_{F}, v_2 \in \cR_{F}, \ldots, v_{j-1} \in \cR_{F}, \mathbb{I}[U_{G'}]}]} \leq \frac{1}{d^{1 - \epsilon}}\leq \frac{2 d^{\epsilon}}{2\alpha}. 
		\end{equation*}
		\kaedit{
		Note that,
		\begin{align*}
			 &\Pr[v_1 \in \cR_{F} \wedge v_2 \in \cR_{F} \wedge \ldots \wedge v_t \in \cR_{F}\given{\mathbb{I}[U_{G'}]}]\\
			 &= \Pr[v_t \in \cR_{F}\given{v_1 \in \cR_{F} \wedge v_2 \in \cR_{F} \wedge \ldots \wedge v_{t-1} \in \cR_{F}, \mathbb{I}[U_{G'}]}] \times \\
			 & \Pr[v_{t-1} \in \cR_{F}\given{v_1 \in \cR_{F} \wedge v_2 \in \cR_{F} \wedge \ldots \wedge v_{t-2} \in \cR_{F}, \mathbb{I}[U_{G'}]}] \times \ldots \times \Pr[v_1 \in \cR_{F}\given{\mathbb{I}[U_{G'}]}]
		\end{align*}
		
		Chaining these together we obtain,
		\[
			\Pr[v_1 \in \cR_{F} \wedge v_2 \in \cR_{F} \wedge \ldots \wedge v_t \in \cR_{F}\given{\mathbb{I}[U_{G'}]}] \leq \left( \tfrac{d^{\epsilon}}{\alpha} \right)^{t-1} \tfrac{1}{2\alpha}.
		\]
		
		Combining this with Equation~\ref{eqn:all-in-G'} we get,
		\[
			\Pr[v_1 \in \cR_{F} \wedge v_2 \in \cR_{F} \wedge \ldots \wedge v_t \in \cR_{F}] \leq (2 d^{\epsilon})^{t-1} \cdot \prod_{j = 1}^t \frac{x_{v_j}}{2k}. \qedhere
		\]
		}
\end{proof}
\kaedit{
\begin{remark}
	\label{rem:concNeg}
	Let $X_j$ denote the indicator random variable for the event that $j \in \mathcal{R}_F$. Then, by the results of \cite{PS97,SSS95}, Theorem~\ref{thm:negativeCorrelation} yields upper-tail bounds for any non-negative linear combination of the $X_j$'s.
\end{remark}
}

	\xhdr{Implementation in the distributed setting.} We briefly give a high-level description on how to obtain distributed algorithms with the same approximation ratio for \kcspip. The algorithm described above cannot directly be implemented in the distributed model of computation. We now briefly describe how our algorithm can be modified to overcome this. Note that step~\ref{ln:color} in the current algorithm makes the algorithm inherently \emph{sequential}. \kaedit{The running time of our algorithm is determined by this step and runs in time $O(\poly n)$. However, we can obtain the coloring by using known distributed algorithms that obtain the same ratio and runs in time $O(k^{0.4} \poly \log n \poly \log k)$. In particular, we prove the following theorem.
	
		\begin{theorem}[Implementation in the distributed setting]
			\label{thm:distkcspip}
				There exists a rounding algorithm for the \kcspip problem in the distributed setting, that achieves an approximation ratio of $2k + o(k)$ in time $\tilde{O}(k^{0.4} \log n)$, where $\tilde{O}(.)$ hides $\poly \log k$ factors.
		\end{theorem}

		Before we prove the theorem, we recall the following definitions from graph theory.}
		
		\begin{definition}[pseudo-Forest]
			\label{def:pseudoForest}
				A graph is called a \emph{pseudo-forest} if every connected component has \emph{at most} one cycle.
		\end{definition}

		\begin{definition}[Arboricity]
			\label{def:arboricity}
				\emph{Arboricity} of a graph is the minimum number of forests that the edge set of the graph can be partitioned into.
		\end{definition}
		
		\begin{definition}[Pseudo-Arboricity]
			\label{def:pseudoArboricity}
				\emph{Pseudo-Arboricity} of a graph is the minimum number of pseudo-forests that the edge set of the graph can be partitioned into.
		\end{definition}
				
		\begin{definition}[Minimum out-degree orientation]
			\label{def:minOutdegreeOrientation}
				Given an undirected graph $G=(V, E)$ an orientation maps the set of undirected edges to a set of directed edges between the same set of vertices. The out-degree of a orientation is defined as the largest out-degree of any vertex under a given orientation. The \emph{minimum out-degree orientation} is defined as the orientation in which the out-degree is minimum across all possible orientations.
		\end{definition}
		
		\kaedit{
			We now prove Theorem~\ref{thm:distkcspip}.
			\begin{proof}
				Note that both the sampling step and the discarding low-probability events step in Algorithm~\ref{alg:kcspip} can be implemented in parallel in $O(1)$ time. Thus, the bottleneck step in the algorithm is the coloring step which is inherently sequential. Using results from prior-work we show that this can be parallelized as follows. Let $A$ denote the arboricity of the constructed directed graph in step \ref{ln:discard2} of Algorithm~\ref{alg:kcspip}.  Using the Arb-color algorithm from \cite{BE10}\footnote{Note, this algorithm can be implemented without knowing $A$.} with $\epsilon = \frac{1}{d+1}$, we obtain a $2A + 1 + \tfrac{A}{d+1}$ coloring to this directed graph in time $O(A \log n)$. It is known that (see \cite{NW64} for instance) the arboricity of a graph $A$ is at most one more than the pseudo-arboricity of that graph (denote by $\operatorname{PA}$). Thus, the number of colors used is $2\operatorname{PA} + 2 + \tfrac{\operatorname{PA}+1}{d+1}$. Note that the max out-degree in our constructed directed graph is at \talgedit{most} $d$; thus the max-degree in the minimum out-degree orientation is at most $d$. From theorems in \cite{FG78} and \cite{PQ82}, it follows that the pseudo-arboricity of a graph is the out-degree in the minimum out-degree orientation. Therefore, the total number of colors used is at most $2d +3$. Note that $A \leq PA + 1 \leq d + 1$. Thus, the running time of the algorithm is $O(d \log n) = \tilde{O}(k^{0.4} \log n)$. Given a coloring, we can choose an independent set (and thus the rounded items) in $O(1)$ time. Therefore, the entire rounding procedure runs in time $\tilde{O}(k^{0.4} \log n)$.
			\end{proof}
		}

\subsection{Extension to Submodular Objectives}\label{sec:kcs-sub}
				As described in the preliminaries, we can extend certain contention-resolution schemes to submodular objectives using prior work. We will now show that the above rounding scheme can be extended to submodular objectives; in particular, we will use the following definition and theorem from Chekuri~\etal~\cite{CVZ14} \footnote{Specifically, Definition 1.2 and Theorem 1.5 from~\cite{CVZ14}.}.
				
				\begin{definition}
					\label{def:bcCR}
						($bc$-BALANCED MONOTONE CR SCHEMES \cite{CVZ14}) Let $b, c \in [0, 1]$ and let $N:= [n]$ be a set of items. A $bc$-balanced CR scheme $\pi$ for $\cP_{\cI}$ (where $\cP_{\cI}$ denotes the convex relaxation of the set of feasible integral solutions $\cI \subseteq 2^N$) is a procedure which for every $\vec{y} \in b \cP_{\cI}$ and $A \subseteq N$, returns a set $\pi_{\vec{y}}(A) \subseteq A \cap support(\vec{y})$ with the following properties:
						\begin{OneLiners}
								\item[(a)] For all $A \subseteq N, \vec{y} \in b\cP_{\cI}$, we have that  $\pi_{\vec{y}}(A) \in \cI$ with probability 1.
								\item[(b)] For all $i \in support(\vec{y})$, $\vec{y} \in b \cP_{\cI}$, we have that $\Pr[i \in \pi_{\vec{y}}(R(\vec{y})) \given{i \in R(\vec{y})}] \geq c$, where $R(\vec{y})$ is the random set obtained by including every item $i \in N$ independently with probability $y_i$.
								\item[(c)] Whenever $i \in A_1 \subseteq A_2$, we have that $\Pr[i \in \pi_{\vec{y}}(A_1)] \geq \Pr[i \in \pi_{\vec{y}}(A_2)]$.
						\end{OneLiners}
				\end{definition}
				
				\begin{theorem}[Chekuri \etal~\cite{CVZ14}]
					\label{prelim:submodular}
					Suppose $\cP_{\cI}$ admits a $bc$-balanced monotone CR scheme and $\eta_f$ is the approximation ratio for $\max\{F(\vec{x}) : \vec{x} \in \cP_{\cI} \cap \{0, 1\}^n\}$ (here, $F(\vec{x})$ is the multi-linear extension of the sub-modular function $f$). Then there exists a randomized algorithm which gives an expected $1/(bc\eta_f)$-approximation to $\max_{S \in \cI} f(S)$, when $f$ is a non-negative submodular function.
				\end{theorem}
				
				For the case of monotone sub-modular functions, we have the optimal result $\eta_f = 1-1/e$ (Vondr\'ak~\cite{Vo08}). For non-monotone sub-modular functions, the best-known algorithms have $\eta_f \geq 0.372$ due to Ene and Nguyen~\cite{EN16} and more recently $\eta_f \geq 0.385$ due to Buchbinder and Feldman~\cite{BF16} (it is not known if these are tight: the best-known upper bound is $\eta_f \leq 0.478$ due to Oveis Gharan and Vondr\'ak~\cite{GV11}). 

				We will show that Algorithm~\ref{alg:kcspip} is a $1/(2k + o(k))$-balanced monotone CR scheme, for some $b, c$ such that $bc = 1/(2k + o(k))$. Hence, from \cref{prelim:submodular} we have a $(2k + o(k))/\eta_f$-approximation algorithm for \kcspip with sub-modular objectives. This yields Corollary~\ref{cor:kcspip}. 

				\kaedit{\kcspipcor*}
				
				For ease of reading, we will first restate notations used in Definition~\ref{def:bcCR} in the form stated in the previous sub-section. The polytope $\cP_{\cI}$ represents the \kcspip polytope defined by Eq.~\refeq{eq:addcons}. The vector $\vec{y}$ is defined as $y_i:= \alpha x_i/k$ which is used in the sampling step of the algorithm (\ie Step 1). The scheme $\pi_{\vec{y}}$ is the procedure defined by steps \ref{ln:smapling}, \ref{ln:discard1}, \ref{ln:alteration} of the algorithm. In other words, this procedure takes \kaedit{any} subset $A$ of items and returns a \emph{feasible} solution with probability 1 (and hence satisfying property (a) in the definition). Our goal then is to show that it further satisfies properties (b) and (c). 
				
				The set $R(\vec{y})$ corresponds to the set $\cR_0$, where every item $i$ is included into $\cR_0$ with probability $y_i$, independently. From the sampling step of the algorithm, we have that $b=\alpha/k$, since each item $i$ is included in the set $R(\vec{y})$ with probability $y_i:= x_i \alpha /k$ and $\vec{x} \in \cP_{\cI}$ and hence $\vec{y} \in (\alpha/k) \cP_{\cI}$. From the alteration steps we have that $c = (1-o(1))/(2\alpha + o(\alpha))$, since for any item $i$, we have $\Pr[i \in \cR_{F}\given{i \in \cR_0}] \geq  (1-o(1))/(2\alpha + o(\alpha))$. Thus, $\pi_{\vec{y}}$ satisfies property (b). 
				
				Now we will show that the rounding scheme $\pi_{\vec{y}}$ satisfies property (c) in Definition~\ref{def:bcCR}. Let $A_1$ and $A_2$ be two arbitrary subsets such that $A_1 \subseteq A_2$. Consider a $j \in A_1$. We will now prove the following.
					\begin{equation}
						\label{eq:propcsubmodular}	
						\Pr[j \in \RF \given{\cR_0 = A_1}] \geq \Pr[j \in \RF \given{\cR_0=A_2}].
					\end{equation}
					Note that for $i \in \{1, 2\}$ we have,
				\begin{equation}
				\label{eq:submodular}
				 \Pr[j \in \RF \given{\cR_0 = A_i}] = \Pr[j \in \RF \given{j \in \cR_2, \cR_0 = A_i}] \Pr[j \in \cR_2 \given{\cR_0 = A_i}].
				 \end{equation}
				  For both $i=1$ and $i=2$, the first term in the RHS of Eq.~\eqref{eq:submodular} (\ie $\Pr[j \in \RF \given{j \in \cR_2, \cR_0 = A_i}]$) is same and is equal to $1/(2d + 1)$. Note that the second term in the RHS of Eq.~\eqref{eq:submodular} (\ie $\Pr[j \in \cR_2 \given{\cR_0 = A_i}]$) can be rewritten as $\Pr[j \in \cR_2 \given{j \in \cR_0, \cR_0 = A_i}]$ since $j \in A_1$ and $\cR_0 = A_i$ for $i \in \{1, 2\}$. From steps \ref{ln:discard1} and \ref{ln:alteration}\ref{ln:discard2} of the algorithm we have that the event $j \in \cR_2$ conditioned on $j \in \cR_0$ occurs if and only if: 
				  \begin{enumerate}[label=(\roman*)]
				  	\item no medium or tiny blocking events occurred for $j$.
				  	\item vertex $j$ did not correspond to an anomalous vertex in $G$ (one with out-degree greater than $d:=\alpha + \sqrt{\alpha \log \alpha}$).
				  \end{enumerate}
				 Both (i) and (ii) are monotonically decreasing in the set $\cR_0$. \kaedit{In other words, if both the conditions satisfy for an item $j$ when $\cR_0=A_2$, then it also hold when $\cR_0=A_1$. This implies that $\Pr[j \in \cR_2 \given{\cR_0 = A_1}] \geq \Pr[j \in \cR_2 \given{\cR_0 = A_2}]$. Thus combining this with Equation~\eqref{eq:submodular} we obtain Equation~\eqref{eq:propcsubmodular}.}

\section{The Stochastic $k$-Set Packing Problem}
\label{sec:sksp}
Consider the stochastic $k$-set packing problem defined in the introduction. \kaedit{Recall that in this problem, the columns of the packing constraints (\ie size of items) are realized by a stochastic process on probing, and the goal is to choose the optimal order in which the columns (\ie items) have to probed.} We start with a second-chance-based algorithm yielding an improved ratio of $8k/5 + o(k)$. We then improve this to $k + o(k)$ via multiple chances. Recall that if we probe an item $j$, we have to add it irrevocably, as is standard in such stochastic-optimization problems; thus, we do not get multiple opportunities to examine $j$. Let $\vec{x}$ be an optimal solution to the benchmark \LP (\ref{eqn:sksp-lp}) and $\cC(j)$ be the set of constraints that $j$ participates in. 

Bansal \etal~\cite{BGLMNR12} presented Algorithm~\ref{alg:basic-k-set}, $\func{SKSP}(\alpha)$. They show that $\func{SKSP}(\alpha)$ will add each item $j$ with probability at least $\beta \Paren{x_j/k}$ where $\beta \geq \alpha(1-\alpha/2)$. By choosing $\alpha=1$, $\func{SKSP}(\alpha)$ yields a ratio of $2k$.\footnote{The terminology used in~\cite{BGLMNR12} is actually $1/\alpha$; however, we find the inverted notation $\alpha$ more natural for our calculations.}

\IncMargin{1em}
\begin{algorithm}[!h]
	\caption{$\func{SKSP}(\alpha)$~\cite{BGLMNR12}}
	\label{alg:basic-k-set}
	\DontPrintSemicolon
	Let $\cR$ denote the set of chosen items which starts out as an empty set. \;\label{ln:sksp1:init}
	For each $j \in [n]$, generate an independent Bernoulli random variable $Y_{j}$ with mean $\alpha x_j/k$. \;\label{ln:sksp1:sampling}
	Choose a uniformly random permutation $\pi$ over $[n]$ and follow $\pi$ to check each item $j$ one-by-one: add $j$ to $\cR$ if and only if $Y_{j}=1$ and $j$ is safe (i.e., each resource $i \in \cC(j)$ has at least one unit of budget available); otherwise skip $j$. \; \label{ln:sksp1:permutation}
	Return \cR as the set of chosen items.\; \label{ln:sksp1:return}
\end{algorithm}
\DecMargin{1em}

At a high level, our second-chance-based algorithm proceeds as follows with parameters $\{\alpha_1, \beta_1, \alpha_2\}$ to be chosen later. During the first chance, we set $\alpha = \alpha_1$ and run $\func{SKSP}(\alpha_1)$. Let $\ad_{1,j}$ denote the event that $j$ is added to $\cR$ in this first chance. From the analysis in~\cite{BGLMNR12}, we have that $\Pr[\ad_{1,j}] \ge \Paren{x_j/k}\alpha_1(1-\alpha_1/2)$. By applying simulation-based attenuation techniques, we can ensure that each item $j$ is added to \cR in the first chance with probability \emph{exactly equal} to $\beta_1x_j/k$ for a certain $\beta_1 \le \alpha_1(1-\alpha_1/2)$ of our choice.\footnote{See footnote \talgedit{10} in Section~\ref{sec:maintech}. Sampling introduces some small sampling error, but this can be made into a lower-order term with high probability. We thus assume for simplicity that such simulation-based approaches give us exact results.} In other words, suppose we obtain an estimate $\hat{E}_{1,j}:= \Pr[\ad_{1,j}]$. When running the original randomized algorithm, whenever $j$ can be added to $\cR$ in the first chance, instead of adding it with probability 1, we add it with probability $(\Paren{x_j/k}\beta_1)/\hat{E}_{1,j}$. 

In the second chance, we set $\alpha = \alpha_2$ and modify $\func{SKSP}(\alpha_2)$ as follows. We generate an independent Bernoulli random variable $Y_{2,j}$ with mean $\alpha_2 x_j/k$ for each $j$; let $Y_{1,j}$ denote the Bernoulli random variable from the first chance. Proceeding in a uniformly random order $\pi_2$, we add $j$ to \cR if and only if $j$ is safe, $Y_{1,j}=0$ and $Y_{2,j}=1$. Algorithm~\ref{alg:improved-k-set}, $\func{SKSP}(\alpha_1, \beta_1, \alpha_2)$, gives a formal description.

\IncMargin{1em}
\begin{algorithm}[!h]
	\caption{$\func{SKSP}(\alpha_1, \beta_1, \alpha_2)$}
	\label{alg:improved-k-set}
	\DontPrintSemicolon
	Initialize \cR as the empty set. \; \label{ln:sksp2:init}
	{\bf The first chance:}  Run $\func{SKSP}(\alpha_1)$ with simulation-based attenuation such that $\Pr[\ad_{1,j}]=\beta_{1}x_j/k$ for each $j \in [n]$, with $\beta_1 \le \alpha_1(1-\alpha_1/2)$. \cR now denotes the set of variables chosen in this chance. \; \label{ln:sksp2:first}
	{\bf The second chance:} Generate an independent Bernoulli random variable $Y_{2,j}$ with mean $\alpha_2 x_j/k$ for each $j$. Follow a uniformly random order $\pi_2$ over $[n]$ to check each item $j$ one-by-one: add $j$ to \cR if and only if $j$ is safe, $Y_{1,j}=0$, and $Y_{2,j}=1$; otherwise, skip it.\; \label{ln:sksp2:second}
	Return \cR as the set of chosen items.\; \label{ln:sksp2:return}
\end{algorithm}
\DecMargin{1em}

Lemma~\ref{lem:sksp-2} lower bounds the probability that an item gets added in the second chance. For each $j$, let $\ad_{2,j}$ be the event that $j$ is added to \cR in the second chance.

\begin{lemma}
	\label{lem:sksp-2}
	After running $\func{SKSP}(\alpha_1, \beta_1, \alpha_2)$ on an optimal solution $\vec{x}$ to the benchmark \LP (\ref{eqn:sksp-lp}), we have
	\[
	\textstyle \Pr[\ad_{2,j}] \ge \frac{x_j}{k} \alpha_2 \Big( 1-\frac{\alpha_1 x_j}{k}-\beta_1-\frac{\alpha_2}{2} \Big).
	\]
\end{lemma}
\begin{proof}
Let us fix $j$. Note that ``$Y_{1,j}=0$ and $Y_{2,j}=1$" occurs with probability $(1-\alpha_1 x_j/k)\Paren{\alpha_2 x_j/k}$. Consider a given $i \in \cC(j)$ and let $U_{2,i}$ be the budget usage of resource $i$ when the algorithm reaches $j$ in the random permutation of the second chance. 
	\begin{align}
	\Pr[\ad_{2,j}] & =\Pr[Y_{1,j}=0 \wedge Y_{2,j}=1]\Pr[\ad_{2,j} \given{Y_{1,j}=0 \wedge Y_{2,j}=1}]\\
		& =\Paren{1-\frac{\alpha_1 x_j}{k}}\frac{\alpha_2 x_j}{k} \Pr\Bracket{\bigwedge_{i}(U_{2,i} \le b_i-1)\given{Y_{1,j}=0 \wedge Y_{2,j}=1}}\\
		& \ge \Paren{1-\frac{\alpha_1 x_j}{k}}\frac{\alpha_2 x_j}{k} \left(1-\sum_i \Pr\Bracket{U_{2,i} \ge b_i\given{Y_{1,j}=0}}  \right) \\
	& \label{ineq:k-set-two} \ge\frac{\alpha_2 x_j}{k}  \left(1-\frac{\alpha_1 x_j}{k}-\sum_i \Pr\Bracket{U_{2,i} \ge b_i}  \right). 
	\end{align}
	Let $X_{1,\ell}$ be the indicator random variable showing if $\ell$ is added to \cR in the first chance and $\indicator{2,\ell}$ indicate if item $\ell$ falls before $j$ in the random order $\pi_2$. Thus we have,
	\[
	\textstyle U_{2,i} \le \sum_{\ell \neq j} \SI_{i,\ell} \Paren{X_{1,\ell}+(1-Y_{1,\ell})Y_{2,\ell}\indicator{2,\ell}}.
	\]
	which implies
	\begin{equation*}
		\bE[U_{2,i}] \le \sum_{\ell \neq j} u_{i,\ell}\Paren{ \frac{\beta_1 x_\ell}{k}+\frac{1}{2} \Paren{ 1-\frac{\alpha_1 x_\ell}{k} } \frac{\alpha_2 x_\ell}{k} } \le \Paren{ \frac{\beta_1}{k}+\frac{1}{2}\frac{\alpha_2}{k} }b_i.
	\end{equation*}
	Plugging the above inequality into \eqref{ineq:k-set-two} and applying Markov's inequality, we complete the proof of \reflemma{lem:sksp-2}.	
\end{proof}

We can use Lemma~\ref{lem:sksp-2} to show that we get an approximation ratio of $8k/5$. Observe that the events $\ad_{1,j}$ and $\ad_{2,j}$ are mutually exclusive. Let $\ad_j$ be the event that $j$ has been added to \cR after the two chances. Then, by choosing $\alpha_1=1$, $\beta_1=1/2$, and $\alpha_2=1/2$, we have $\Pr[\ad_j]=\Pr[\ad_{1,j}]+\Pr[\ad_{2,j}] = \frac{5 x_j}{8k} - \frac{x_j^2}{2k^2}$. From the Linearity of Expectation, we get that the total expected weight of the solution is at least $(5/8k - o(k))\vec{w}.\vec{x}$.
\begin{theorem}
	\label{thm:sksp2}
	By choosing $\alpha_1=1$, $\beta_1=1/2$, and $\alpha_2=1/2$, $\func{SKSP}(\alpha_1, \beta_1, \alpha_2)$ achieves a ratio of $\frac{8k}{5}+o(k)$ for $\sksp$.
\end{theorem}
	\begin{proof}
		Consider a given item $j$. We have,
	\begin{align*}
	\textstyle \Pr[\ad_j] &\textstyle =\Pr[\ad_{1,j}]+\Pr[\ad_{2,j}] \\
	\textstyle & \textstyle \ge \frac{x_j}{k} \left( \beta_1+ \alpha_2 \Paren{ 1-\frac{\alpha_1 x_j}{k}-\beta_1-\frac{\alpha_2}{2} } \right)\\
	\textstyle & \textstyle \kaedit{\ge} \frac{x_j}{k} \left( \beta_1+ \alpha_2 \Paren{ 1-\beta_1-\frac{\alpha_2}{2} }-O(1/k) \right).
	\end{align*}
	To obtain the worst case, we solve the following optimization problem.
	\begin{equation}\label{eqn:sksp-2}
	 \max~\beta_1+ \alpha_2 \Paren{ 1-\beta_1-\frac{\alpha_2}{2} }, ~\mbox{s.t.}~0 \le \beta_1 \le \alpha_1(1-\alpha_1/2), \alpha_1 \ge 0, \alpha_2 \ge 0.
	\end{equation}
	Solving the above program, the optimal solution is $\alpha_1=1, \beta_1=1/2$ and $\alpha_2=1/2$ with a ratio of $(\frac{8}{5}+o(1))k$.
	\end{proof}

\subsection{Extension to $T$ Chances}
Intuitively, we can further improve the ratio by performing a third-chance probing and beyond. We present a natural generalization of $\talgedit{\func{SKSP}}(\alpha_1, \beta_1, \alpha_2)$ to $\talgedit{\func{SKSP}}(\{\alpha_t, \beta_t\given{t \in [T]}\})$ with $T$ chances, where $\{\alpha_t, \beta_t\given{t \in [T]}\}$ are parameters to be fixed later. Note that $\talgedit{\func{SKSP}}(\alpha_1, \beta_1, \alpha_2)$ is the special case wherein $T=2$. 

During each chance $t \le T$, we generate an independent \emph{Bernoulli} random variable $Y_{t,j}$ with mean $\alpha_t x_j/k$ for each $j$.  Then we follow a uniform random order $\pi_t$ over $[n]$ to check each item $j$ one-by-one: we add $j$ to \cR if and only if $j$ is safe, $Y_{t',j}=0$ for all $t'<t$ and $Y_{t,j}=1$; otherwise we skip it. Suppose for a chance $t$, we have that each $j$ is added to \cR with probability at least $\beta_t x_j/k$. As before, we can apply simulation-based attenuation to ensure that each $j$ is added to \cR in chance $t$ with probability exactly equal to $\beta_t x_j/k$. To achieve this goal we need to simulate our algorithm over all previous chances up to the current one $t$. Algorithm~\ref{alg:sksp-T}, $\func{SKSP}(\{\alpha_t, \beta_t| t \in [T]\})$, gives a formal description of the algorithm. Notice that during the last chance $T$, we do not need to perform simulation-based attenuation. For the sake of uniformity in presentation, we still describe it in the algorithm description.

\IncMargin{1em}
\begin{algorithm}[!h] 
	\caption{$\func{SKSP}(\{\alpha_t, \beta_t| t \in [T]\})$}
	\label{alg:sksp-T}
	\DontPrintSemicolon
	Initialize \cR as the empty set. \; \label{ln:sksp3:init}
	\For{t=1, 2, \ldots, T}{
	\talgedit{	Generate an independent Bernoulli random variable $Y_{t,j}$ with mean $\alpha_t x_j/k$ for each $j$. \; \label{ln:sksp3:permutation}
		Apply simulation-based attenuation such that for each $j$ is added to \cR in the $t^{th}$ chance (denote by an indicator random variable $Z_{t, j}$) with probability equal to $\beta_t x_j/k$.\; \label{ln:sksp3:attenuation}
	Follow a uniform random order $\pi_t$ over $[n]$ to check each item $j$ one-by-one: add $j$ to \cR if and only if $j$ is safe, $Z_{t',j}=0$ for all $t'<t$, and $Z_{t,j}=1$; otherwise, skip it.\; }
	}
	Return \cR as the set of chosen items.\; \label{ln:sksp3:return}
\end{algorithm}
\DecMargin{1em}


For each item $j$, let $\ad'_{t,j}$ be the event that $j$ is added to \cR in the $t^{th}$ chance before step \ref{ln:sksp3:permutation} of the algorithm for $t$ (\ie before the start of the \talgedit{$(t+1)^{th}$} iteration of the loop). Lemma~\ref{lem:sksp-T} lower bounds the probabilities of these events.

\begin{lemma}
\label{lem:sksp-T}
After running $\talgedit{\func{SKSP}}(\{\alpha_t, \beta_t\given{t \in [T]}\})$ on an optimal solution $\vec{x}$ to the benchmark \LP (\ref{eqn:sksp-lp}), we have
	\[
	\textstyle \Pr[\ad'_{t,j}]  \ge  \frac{x_j}{k}\left(\alpha_t \Big(1 -\sum_{t'<t} \beta_{t'}-\frac{\alpha_t}{2}\Big)-\frac{\alpha_t\sum_{t'<t} \alpha_{t'}}{k}\right).
	\]
\end{lemma}
\begin{proof}
Consider an item $j$ and define \kaedit{indicator} random variable $Z_{t, j}=1$ iff $Y_{t',j}=0$ for all $t'<t$ and $Y_{t, j}=1$. Observe that $\bE[Z_{t, j}] =\frac{\alpha_t x_j}{k}\prod_{t'<t} \Paren{1-\frac{\alpha_{t'} x_j}{k}}$. Consider a given $i \in \cC(j)$ and let $U_{t,j}$ be the budget usage of resource $i$ when the algorithm reaches $j$ in the random permutation during chance $t$. Thus we have,
	\begin{align*}
	&\textstyle \Pr[\ad'_{t,j}] =\Pr[Z_{t,j}=1]\Pr[\ad'_{j,t}\given{Z_{t,j}=1}]\\
	&\textstyle =\frac{\alpha_t x_j}{k}\prod_{t'<t} \Paren{1-\frac{\alpha_{t'} x_j}{k}} \\
	& \textstyle \hspace{10mm} \Pr\Bracket{\bigwedge_{i}(U_{t,i} \le b_i-1)\given{Z_{t,j}=1}} \\
	& \textstyle \ge \frac{\alpha_t x_j}{k}\prod_{t'<t} \Paren{1-\frac{\alpha_{t'} x_j}{k}} \\
	& \textstyle \hspace{10mm}  \left(1-\sum_i \Pr\Bracket{U_{t,i} \ge b_i \given{\bigwedge_{t'<t}Y_{t',j}=0}}  \right)  \\
	& \textstyle \ge\frac{\alpha_t x_j}{k}  \left(\prod_{t'<t} \Paren{1-\frac{\alpha_{t'} x_j}{k}}-\sum_i \Pr\Bracket{U_{t,i} \ge b_i}  \right). \label{ineq:k-set-T}
	\end{align*}
	Notice that 
	\begin{equation*}
	\bE[U_{t,j}] \le  \sum_{\ell \neq j} u_{i,\ell} \left( \frac{ x_\ell}{k} \sum_{t'<t} \beta_{t'}+ \frac{1}{2} \prod_{t'<t}\Paren{1-\frac{\alpha_{t'} x_\ell}{k}} \frac{\alpha_t x_\ell}{k} \right) \le \left(\frac{\sum_{t'<t}\beta_{t'}}{k}+\frac{1}{2}\frac{\alpha_t}{k}\right)b_i.
	\end{equation*}
	By applying Markov's inequality, we get
	\begin{align*}
	\Pr[\ad'_{t,j}]  & \textstyle \ge \frac{x_j}{k}\alpha_t\left(\prod_{t'<t} \Paren{1-\frac{\alpha_{t'} x_j}{k}} -\Paren{\sum_{t'<t} \beta_{t'}+\frac{\alpha_t}{2}} \right)\\
	& \ge  \frac{x_j}{k}\left(\alpha_t \Paren{1 -\sum_{t'<t} \beta_{t'}-\frac{\alpha_t}{2}}-\frac{\alpha_t\sum_{t'<t} \alpha_{t'}}{k}\right). \qedhere
	\end{align*}	
\end{proof}

Combining Lemma~\ref{lem:sksp-T} and simulation-based attenuation, we have that for any given $\{\alpha_{t'} \given{t' \le t}\}$ and $\{\beta_{t'}\given{t'<t}\}$, each item $j$ is added to \cR in chance $t$ with probability equal to $\beta_t x_j/k$ for any \\
$\beta_t \le \alpha_t \Paren{1 -\sum_{t'<t} \beta_{t'}-\alpha_t/2 }-\alpha_t\sum_{t'<t} \alpha_{t'}/k $. For each $j$, let $E_{j,t}$ be the event that $j$ is added to \cR in chance $t$ 
and $E_{j}$ the event that $j$ is added to \cR after $T$ chances. From Algorithm~\ref{alg:sksp-T}, we have that the events $\{E_{j,t}\given{t \le T}\}$ are mutually exclusive. Thus, $\Pr[E_j]=\sum_{t \le T}\Pr[E_{j,t}]=\sum_{t \le T} \beta_t x_j/k$. Therefore, to maximize the final ratio, we will solve the following optimization problem.
\begin{equation}\label{eqn:sksp-T}
 \max \sum_{t \in [T]} \beta_t ~\mbox{s.t.}~\beta_t \le \alpha_t \Paren{1 -\sum_{t'<t} \beta_{t'}- \frac{\alpha_t}{2}} -\frac{\alpha_t\sum_{t'<t} \alpha_{t'}}{k} ~ \forall t \in [T], ~\alpha_t \ge 0 ~\forall t \in [T].
\end{equation}
Consider a simplified version of the maximization program~\eqref{eqn:sksp-T} by ignoring the $O(1/k)$ term as follows.
\begin{equation}\label{eqn:sksp-TT}
	\max \sum_{t \in [T]} \beta_t, ~\mbox{s.t.}~\beta_t \le \alpha_t \Paren{1 -\sum_{t'<t} \beta_{t'}-\frac{\alpha_t}{2}} ~\forall t \in [T], ~\alpha_t \ge 0 ~\forall t \in [T]. 
\end{equation}

\begin{lemma}\label{lem:sksp-}
	An optimal solution to the program~\eqref{eqn:sksp-TT} is
	\[
		 \beta^*_t=\frac{1}{2}\left( 1-\sum_{ t'<t} \beta^*_{t'} \right)^2, ~\forall t\ge 1, \alpha^*_t=1-\sum_{ t'<t}\beta^*_{t'},~ \forall t\ge 1.
	\]
	where $\beta^*_0=0$, $0 \le \alpha^*_t \le 1$\s for all $t \ge 1$ and $\lim_{T \rightarrow \infty}\sum_{t=1}^{T} \beta^*_t=1$. 
\end{lemma}
\begin{proof}
For any given $\{\beta_t\given{1 \le t<T}\}$, we have $\beta_T \le \alpha_T \Paren{1 -\sum_{t'<T} \beta_{t'}-\alpha_T/2}$. Thus, in the optimal solution we have that $\beta_T^*=\frac{1}{2}(1-\sum_{ 1\le t'<T} \beta^*_{t'})^2$ and $\alpha_T=1-\sum_{1 \le t'<T}\beta^*_{t'}$. Thus,
	\[
	\textstyle \sum_{t \in [T]} \beta_t=\sum_{t<T}\beta_t+\beta_T^*=\frac{1}{2}\Paren{1+\big(\ts\sum_{t<T} \beta_t\big)^2 }.
	\]
	Therefore the maximization of $\sum_{t \le T}\beta_t$ is reduced to that of $\sum_{t \le T-1} \beta_t$. Repeating the above analysis $T$ times, we get our claim for $\beta^*_t$ and $\alpha_t^*$. Let $\gamma^*_T=\sum_{t \in [T]} \beta_t^*$. From the above analysis, we have
	\[
	\textstyle
	\gamma^*_1=\frac{1}{2},~~\gamma^*_{t}=\frac{1}{2}\Big(1+(\gamma_{t-1}^*)^2\Big),~\forall t \ge 2.
	\]
	Since $\gamma^*_1 \le 1$, we can prove that $\gamma^*_t \le 1$ for all $t$ by induction. Notice that $\gamma^*_{t}-\gamma^*_{t-1}=\frac{1}{2}(1-\gamma^*_{t-1})^2 \ge 0$, which implies that $\{\gamma_t\}$ is a non-decreasing. Since $\{\gamma_t\}$ is non-decreasing and bounded, it has a limit $\ell$. The only solution to the equation $\ell = (1 + \ell^2)/2$ is $\ell = 1$, and hence $\lim_{T \rightarrow \infty} \gamma^*_T=1$. 	
\end{proof}


\begin{theorem}
	\label{thm:skspthm2}
	Let $T$ be some slowly-growing function of $k$, e.g., $T = \log k$. For each $t \in [T]$, set $\bar{\alpha}_t=\alpha^*_t, ~\bar{\beta_t}= \beta^*_{t}-\frac{\alpha_t^*\Paren{\sum_{t'<t} \alpha^*_{t'}}}{k}$. Then we have (1) $\{\balp_t, \bbeta_t\given{t \in [T]}\}$ is feasible for the program~\eqref{eqn:sksp-T} and (2) $\sum_{t \in [T]} \bbeta_t=1+o(1)$ where $o(1)$ goes to $0$ when $k$ goes to infinity. Thus, $\func{SKSP}(\{\bar{\alpha}_t, \bar{\beta}_t| t \in [T]\})$ achieves a ratio of $k + o(k)$ for \sksp.
\end{theorem}
\begin{proof}
Observe that for each $t \in [T]$, $\bbeta_t \le \beta_t^*$. Also,
	\begin{align*}
	\textstyle \bbeta_t&\textstyle =\frac{1}{2}\Paren{1-\sum_{t'<t} \beta^*_{t'}}^2-\frac{\alpha_t^*\Paren{\sum_{t'<t} \alpha^*_{t'}}}{k}\\
	&\textstyle = \alpha^*_t \Paren{1 -\sum_{t'<t} \beta^*_{t'}-\frac{\alpha^*_t}{2}}-\frac{\alpha_t^*\Paren{\sum_{t'<t} \alpha^*_{t'}}}{k}  \\
	&\textstyle \le  \balp_t \Paren{1 -\sum_{t'<t} \bbeta_{t'}-\frac{\balp_t}{2}}-\frac{\balp_t\Paren{\sum_{t'<t} \balp_{t'}}}{k}.
	\end{align*}
	Thus, we claim that $\{\balp_t, \bbeta_t\given{t \in [T]}\}$ is feasible for the program~\eqref{eqn:sksp-T}. Notice that 
	\begin{equation*}
	 \sum_{t \in [T]} \bbeta_t=\sum_{t \in [T]}\beta^*_t-\sum_{t \in [T]}\frac{\alpha_t^*\big(\sum_{t'<t} \alpha^*_{t'}\big)}{k} \ge \gamma^*_T-\frac{T^2}{k} =1-(1-\gamma^*_T+\frac{\log^2 k}{k}).
	\end{equation*}
	From Lemma~\ref{lem:sksp-}, we have that $(1-\gamma^*_T)=o(1)$ thus proving the theorem.	
\end{proof}

\section{Hypergraph Matching}
\label{sec:hypermatch}

In this section, we give a non-uniform attenuation approach to the hypergraph matching problem which leads to improved competitive ratios. Additionally as stated in the introduction, this takes us ``half the remaining distance'' towards resolving the stronger Conjecture~\ref{conj:fks-ours}.

Consider a hypergraph $\cH=(\cV, \cE)$. Assume each $e \in \cE$ has  cardinality $|e|=k_e$. Let $\vec{x}=\{x_e\}$ be an optimal solution to the \LP (\ref{eqn:hm-lp}). We start with a warm-up algorithm due to Bansal \etal~\cite{BGLMNR12}\footnote{Similar to our algorithm for \sksp, we use the notation $\alpha$ while \cite{BGLMNR12} use $1/\alpha$.}. \kaedit{Note that \cite{BGLMNR12} study the more general $\sksp$ problem. The algorithm they describe is the one described in Algorithm~\ref{alg:basic-k-set} in the previous section. Here, we show that when restricted to the special setting of hypergraph matching, the same algorithm yields an approximation ratio of $k+1 + o(k)$. We use this as a starting point to obtain an algorithm with improved competitive ratio. We recall their approach, denoted by $\HM(\alpha)$, in Algorithm~\ref{alg:hm-basic}.}

\IncMargin{1em}
\begin{algorithm}[!h]
\caption{$\HM(\alpha)$}
\label{alg:hm-basic}
\DontPrintSemicolon
Initialize \cR to be the empty set. We will add edges to this set during the algorithm and return it at the end as the matching. \; \label{ln:hm1:init}
For each $e \in \cE$, generate an independent Bernoulli random variable $Y_{e}$ with mean $\alpha x_e$. \; \label{ln:hm1:sampling}
\kaedit{For every edge $e \in \cE$ with $Y_e=1$, choose a random number $x_e \in [0, 1]$ uniformly at random and independent of the other edges. Let $\pi$ denote the ordering over $\cE$ such that the realized values $(x_e)_{e \in \cE}$ is sorted in ascending order.} Follow $\pi$ to check each edge one-by-one: add $e$ to \cR if and only if $Y_{e}=1$ and $e$ is safe (i.e., none of the vertices in $e$ are matched); otherwise skip it. \;	\label{ln:hm1:permutation}
Return \cR as the matching. \; \label{ln:hm1:return}
\end{algorithm}
\DecMargin{1em}

\begin{lemma}\label{lem:hm-basic}
Each edge $e$ is added to \cR with probability at least $\frac{x_e}{k_e+1 \kaedit{+ o(k_e)}}$ in $\HM(\alpha)$ with $\alpha=1$.
\end{lemma}
\begin{proof}
The proof is very similar to that in~\cite{BCNSX15} (for instance see Lemma 10 in \cite{BCNSX15}). For each $e\in \cE$, let $\ad_e$ be the event that $e$ is added to \cR in $\HM(\alpha)$. 
\kaedit{
Consider an edge $e$ and for each vertex $v \in e$, let $L_v$ be the event that $v$ is unmatched when considering $e$ in the permutation $\pi$. Let $\pi(e)=t \in (0,1)$. Thus, $\ad_e$ occurs if $Y_e=1$ and every vertex $v \in e$ is unmatched when considering $e$ in the permutation. Therefore, we have the following.
	\begin{align}
	\textstyle \Pr[\ad_e]&\textstyle =\alpha x_e\int_0^1  \Pr\Bracket{ \bigwedge_{v \in e} L_v\given{\pi(e)=t}} dt \nonumber \\
	&\textstyle = \alpha x_e\int_0^1 \prod_{v \in e} \prod_{f: f\ni v}(1-t \alpha x_f) dt \label{eq:HM2}\\
	&  \textstyle \ge \alpha x_e\int_0^1 (1-t \alpha)^{\sum_{v \in e} \sum_{f: f\ni v} x_f} dt \label{eq:HM3}\\ 
	&\textstyle \ge  \alpha x_e\int_0^1 (1-t\alpha)^{k_e} dt \label{eq:HM4} \\
	& \textstyle =\frac{x_e}{k_e+1}\Big(1-(1-\alpha)^{k_e+1} \Big). \nonumber
	\end{align}
	Equation~\eqref{eq:HM2} can be obtained as follows. For a given vertex $v \in e$ to be safe, we want that none of the edges incident to $v$ be matched when considering $e$. This is precisely the negation of the probability that an edge incident to $v$ is matched earlier in the permutation. A similar argument is also made in Lemma 10 of \cite{BCNSX15}. Equation~\eqref{eq:HM3} is obtained as follows. Note that $x_f \leq 1$ and $\alpha=1$. This implies that $t \alpha \leq 1$ when $t \in [0, 1]$. From Lemma~\ref{lem:HMtech} in the appendix we have that $(1-t \alpha x_f) \geq (1-t \alpha)^{x_f}$. Equation~\eqref{eq:HM4} is obtained as follows. From the LP constraints we have that $\sum_{f: f\ni v} x_f \leq 1$. Moreover the number of vertices $v \in e$ is precisely $k_e$. Combining this with the fact that $1-t \alpha \leq 1$, we obtain the inequality. This completes the proof of the lemma.
}
\end{proof}

It can be shown that in $\HM(\alpha=1)$, the worst case occurs for the edges $e$ with $x_e \leq \epsilon \approx 0$ (henceforth referred to as ``tiny'' edges). In contrast, for the edges with $x_e \gg \epsilon$ (henceforth referred to as ``large'' edges), the ratio is much higher than the worst case bound. This motivates us to balance the ratios among tiny and large edges. Hence, we modify Algorithm~\ref{alg:hm-basic} as follows: we generate an independent Bernoulli random variable $Y_e$ with mean $g(x_e)$ for each $e$, where $g: [0,1] \rightarrow [0,1]$ is \kaedit{once differentiable and satisfying $g(0)=0$ and $g'(0)=1$}. Function $g$ will be fixed in the analysis. Algorithm~\ref{alg:hm-f} gives a formal description of this modified algorithm.

\IncMargin{1em}
\begin{algorithm}[!h]
\caption{$\HM(g)$}
\label{alg:hm-f}
\DontPrintSemicolon
Initialize \cR to be the empty set. \; \label{ln:hm2:init}
For each $e \in \cE$, generate an independent Bernoulli random variable $Y_{e}$ with mean $g(x_e)$. \; \label{ln:hm2:sampling}
\kaedit{For every edge $e \in \cE$ with $Y_e=1$, choose a random number $x_e \in [0, 1]$ uniformly at random and independent of the other edges. Let $\pi$ denote the ordering over $\cE$ such that the realized values $(x_e)_{e \in \cE}$ is sorted in ascending order.} Follow $\pi$ to consider each edge one by one: add $e$ to \cR if $Y_{e}=1$ and $e$ is safe; otherwise skip it. \;	\label{ln:hm2:permutation}		
Return $\cR$ as the matching \; \label{ln:hm2:return}
\end{algorithm}
\DecMargin{1em}

Observe that $\HM(\alpha)$ is the special case wherein $g(x_e)=\alpha x_e$. We now consider the task of finding the optimal $g$ such that the resultant ratio achieved by $\HM(g)$ is maximized. Consider a given $e$ with $x_e=x$. For any $e' \neq e$, we say $e'$ is a neighbor of $e$ (denoted by $e' \sim e$) if $e' \ni v$ for some $v \in e$. From the LP  (\ref{eqn:hm-lp}), we have $\sum_{e' \sim e} x_{e'} \le k_e(1-x)$. Let $\ad_{e}$ be the event that $e$ is added to \cR. By applying an analysis similar to the proof of Lemma~\ref{lem:hm-basic}, we get the probability of $\ad_{e}$ is at least
\begin{equation}
\label{fks:main-rhs}
\textstyle \Pr[\ad_{e}] \ge g(x) \int_0^1 \prod_{e' \sim e} (1-t g(x_{e'})) dt.
\end{equation}
Therefore our task of finding an optimal $g$ to maximize the r.h.s.\ of (\ref{fks:main-rhs}) is equivalent to finding $\max_g \cF(g)$, where $\cF(g)$ is defined in equation \eqref{eqn:hm-main}.
\begin{equation}\label{eqn:hm-main}
\textstyle \cF(g) \doteq \kaedit{\min_{x \in (0,1)}}
 \left[\frac{g(x)}{x} \times  r(x)\right].
\end{equation}
In equation~\eqref{eqn:hm-main}, $r(x)$ is defined as 
\begin{equation*}
\textstyle r(x) \doteq \min \int_0^1 \prod_{e' \sim e} \big(1-t g(x_{e'})\big) dt, \textstyle ~\text{where} \sum_{e' \sim e} x_{e'} \le k_e(1-x), x_{e'} \in [0,1], \forall e'. 
\end{equation*}
\vspace{-3mm}
\begin{lemma}\label{lem:hm-cF}
By choosing $g (x)= x(1-\frac{x}{2})$, we have that the minimum value of $\cF(g)$ in Eq.~\eqref{eqn:hm-main} is $\cF(g)=\frac{1}{k_e}(1-\exp(-k_e))$.
\end{lemma}
\begin{proof}
\kaedit{Note that $g(0)= 0$. Additionally, $g'(x) = 1-x$ and thus $g'(0) = 1$.}

 Consider a given $x_e=x$ with $|e|=k_e$. Notice that for each given $t \in (0,1)$,
	\[
	\textstyle \prod_{e' \sim e} \big(1-t g(x_{e'})\big) =\exp\Big( \sum_{e' \sim e} \ln(1-t g(x_{e'}))\Big).
	\]
	Note that $g(x) =x(1-x/2)$ satisfies the condition of Lemma~\ref{lem:hm-convex} in the Appendix and hence, for each given $t \in (0,1)$, the function $\ln(1-t g(x))$ is convex over $x \in [0,1]$. Thus to minimize $ \sum_{e' \sim e} \ln(1-t g(x_{e'}))$ subject to $0 \le x_{e'} \le 1$ and $\sum_{e'} x_{e'} \le \kappa$ with $\kappa=k_e(1-x)$, an adversary will choose the following worst case scenario: create $\kappa/\epsilon$ neighbors for $e$ with each $x_{e'}=\epsilon$ and let $ \epsilon\rightarrow 0$. Thus,
	\begin{equation*}
 	\min \prod_{e' \sim e} \big(1-t g(x_{e'})\big) =\min \exp\Big( \sum_{e' \sim e} \ln(1-t g(x_{e'}))\Big) = \lim_{\epsilon \rightarrow 0} (1-tg(\epsilon))^{\kappa/\epsilon}=\exp(-t\kappa).
	\end{equation*}
	\kaedit{The last inequality is obtained as follows. Let $y := \lim_{\epsilon \rightarrow 0} (1-tg(\epsilon))^{\kappa/\epsilon}$. Taking $\ln$ on both sides, we obtain $\ln y = \lim_{\epsilon \rightarrow 0} \tfrac{\kappa}{\epsilon} \ln (1-t g(\epsilon) )$. Using the L'Hopital's rule, we obtain $\ln y = -\kappa t \left( \lim_{\epsilon \rightarrow 0} \tfrac{g'(\epsilon)}{1-tg(\epsilon)} \right)$. Since $g'(0) = 1$ and $g(0) = 0$, thus the limit evaluates to $-\kappa t$. Taking exponentials on both sides, we obtain $y = \exp(-t \kappa)$.
	}
	Therefore, for each fixed $x_e=x$, the optimal value to the inner minimization program in (\ref{eqn:hm-main}) has the following analytic form.
	\begin{equation}\label{eqn:hm-main-proof}
	\min \int_0^1 \prod_{e' \sim e} \big(1-t g(x_{e'})\big) dt= \int_0^1 \exp(-t\kappa) dt=
	\frac{1-\exp(-\kappa)}{\kappa}.
	\end{equation}
	Plugging this back into (\ref{eqn:hm-main}), we obtain
	$\cF(g)=\min_{x \in [0,1]} G(x)$, where
	\[
	\textstyle G(x) \doteq \Big(1-\frac{x}{2}\Big) \frac{1}{k_e(1-x)}\Big(1-\exp(-k_e(1-x)) \Big).
	\]
	Note that $G'(x) \ge 0$ whenever $x\in [0,1]$ and thus, the minimum value of $G(x)$ in the range $x \in (0,1)$ is $G(0) = \frac{1}{k_e}(1-\exp(-k_e))$. \kaedit{Moreover, for the class of functions $g$ with $g(0)=0$ and $g'(0)=1$, this analysis is tight.}
\end{proof}

We now prove the main result, Theorem~\ref{thm:hm}.

\begin{proof}
Consider $\HM(g)$ as shown in Algorithm~\ref{alg:hm-f} with $g(x)=x(1-x/2)$. Let $\cR$ be the random matching returned. From Lemma~\ref{lem:hm-cF}, we have that each $e$ will be added to \cR with probability at least $x_e \cF(g)=\frac{x_e}{k_e}(1-\exp(-k_e))$. 
\end{proof}

\section{More Applications}
\label{sec:otherapp}
		In this section, we briefly describe how a simple simulation-based attenuation can lead to improved contention resolution schemes for \ufp with unit demands. This version of the  problem was studied by Chekuri \etal~\cite{CMS07} where they gave a $4$-approximation for the linear objective case. They also described a simple randomized algorithm that obtains a  $27$-approximation\footnote{This can be obtained by maximizing over all $0  < b < 1/3e$ in Lemma 4.19 of~\cite{CVZ14}, which yields approximately $1/27$.}. Later, Chekuri \etal~\cite{CVZ14} developed the machinery of contention resolution schemes, through which they extended it to a $27/ \eta_f$-approximation algorithm for non-negative submodular objective functions (where $\eta_f$ denotes the approximation ratio for maximizing non-negative submodular functions\footnote{See the section on \kcspip for a discussion on various values of $\eta_f$ known.}). We show that using simple attenuation ideas can further improve this $27/\eta_f$-approximation to an $8.15/\eta_f$-approximation. We achieve this by improving the $1/27$-balanced CR scheme\footnote{See the section on extension to sub-modular objectives in \kcspip for defintion of a balanced CR scheme.} to a $1/8.15$-balanced CR scheme and hence, from Theorem 1.5 of \cite{CVZ14}, the approximation ratio follows.
		
	Consider the natural packing \LP relaxation. Associate a variable $x_i$ with every demand pair. Our constraint set is: for every edge $e$, $\sum_{i: e \in \cP_i} x_i \leq u_e$, where $u_e$ is the capacity of $e$. Our algorithm (formally described in Algorithm~\ref{alg:ufp}) proceeds similar to the one described in \cite{CMS07}, except at line 3, where we use our attenuation ideas.
		
		\IncMargin{1em}
		\begin{algorithm}[!h]
			\caption{Improved contention resolution Scheme \ufp with unit demands}
			\label{alg:ufp}
			\DontPrintSemicolon
				Root the tree $T$ arbitrarily. Let $\lca(s_i, t_i)$ denote the least common ancestor of $s_i$ and $t_i$.\; \label{ln:ufp:init}
				Construct a random set $\cR$ of demand pairs, by including each pair in $\cR$ independently with probability $\alpha x_i$.\; \label{ln:ufp:sampling}
				Consider the demand pairs in increasing distance of their $\lca$ from the root. Let $\cR_{\text{final}}$ denote the set of pairs included in the rounded integral solution. For every demand pair $i$, simulate the run of the algorithm from beginning (\ie produce many samples of $\cR$ and separately run the algorithm on these samples) to obtain the estimate $\eta_i$ of the probability of $i$ to be safe (\ie none of the edges in the path has exhausted capacities). Suppose $i$ is safe, add it to $\cR_{\text{final}}$ with probability $\beta/\eta_i$.\; \label{ln:ufp:attenuation}
				Return $\cR_{\text{final}}$ as the set of demands chosen from routing.\; \label{ln:ufp:return}
		\end{algorithm}
		\DecMargin{1em}
		
		\xhdr{Analysis.} For the most part, the analysis is similar to the exposition in Chekuri \etal~\cite{CVZ14}. We will highlight the part where \emph{attenuation} immediately leads to improved bounds.
		
		Consider a fixed pair $(s_{i^*}, t_{i^*})$ and let $\ell := \lca(s_{i^*}, t_{i^*})$ in $T$. Let $P$ and $P'$ denote the unique path in the tree from $\ell$ to $s_{i^*}$ and $\ell$ to $t_{i^*}$ respectively. As in \cite{CVZ14}, we will upper bound the probability of ${i^*}$ being unsafe due to path $P$ and a symmetric argument holds for $P'$. Let $e_1, e_2, \ldots, e_\lambda$ be the edges in $P$ from $\ell$ to $s_{i^*}$. Let $E_j$ denote the event that ${i^*}$ is not safe to be added in line \ref{ln:ufp:attenuation} of Algorithm~\ref{alg:ufp}, because of overflow at edge $e_j$. Note that for $j>h$ and $u_{e_j} \geq u_{e_h}$, event $E_j$ implies $E_h$ and hence $\Pr[E_j] \leq \Pr[E_h]$. Note, this argument does not change due to attenuation since the demands are processed in increasing order of the depth and any \emph{chosen} demand pair using edge $e_j$ also has to use $e_h$ up until the time ${i^*}$ is considered. Thus, we can make a simplifying assumption similar to \cite{CVZ14} and consider a strictly decreasing sequence of capacities $u_{e_1} > u_{e_2} > \ldots > u_{e_\lambda} \geq 1$. Let $\cS_j$ denote the set of demand pairs that use edge $e_j$. The following steps is the part where our analysis differs from \cite{CVZ14} due to the introduction of attenuation. 
		
		Define, $\beta := 1-2\alpha e/(1- \alpha e)$ and $\gamma := \alpha \beta$. Note that without attenuation, we have $\eta_i \geq \beta$ for all $i$ from the analysis in \cite{CVZ14}.
		
		Let $E'_j$ denote the event that at least $u_{e_j}$ demand pairs out of $\cS_j$ are included in the final solution. Note that $\Pr[E_j] \leq \Pr[E'_j]$. From the LP constraints we have $\sum_{i \in \cS_j} x_i \leq u_{e_j}$. 
		
		Let $X_i$ denote the indicator random variable for the following event: \{$i \in \cR \wedge i \in \cR_{\text{final}}$\}. We define $X := \sum_{i \in \cS_j} X_i$. 
		
		Note that, event $E'_j$ happens if and only if $X \geq u_{e_j}$ and hence we have, $\Pr[E'_j] = \Pr[X \geq u_{e_j}]$.
		Additionally, we have that the $X_i$'s are ``cylindrically negatively correlated" and hence we can apply the Chernoff-Hoeffding bounds due to Panconesi and Srinivasan~\cite{PS97}. Observe that $\bE[X] \leq \gamma u_{e_j}$(since each $i$ is included in $\cR$ independently with $\alpha x_i$ and then included in $\cR_{\text{final}}$ with probability exactly $\beta$) and for $1 + \delta = 1/\gamma$, we have $\Pr[X \geq u_{e_j}] \leq (e/(1+\delta))^{(1+\delta)\mu} \leq (\gamma e)^{u_{e_j}}$. Hence, taking a union bound over all the edges in the path, we have the probability of ${i^*}$ being unsafe due to an edge in $P$ to be at most $\sum_{q=1}^{\infty} (\gamma e)^{\ell} = (\gamma e)/(1-\gamma e)$ (We used the fact that $u_{e_1} > u_{e_2} > \ldots > u_{e_\lambda} \geq 1$). Combining the symmetric analysis for the other path $P'$, we have the probability of ${i^*}$ being unsafe to be at most $2\gamma e/(1-\gamma e)$. Note that we used the fact that $\gamma e < 1$ in the geometric series. Additionally, since $\gamma \leq 1$, we have that $2\gamma e/(1-\gamma e) \leq 2 \alpha e/(1-\alpha e)$. Hence, using $\eta_i \geq \beta$ is justified. 
		
		Now to get the claimed approximation ratio, we solve the following maximization problem: 
		\[
			\max_{0 \leq \alpha \leq 1} \{ \alpha \cdot (1-2\gamma e/(1-\gamma e)) : \beta=1-2\alpha e/(1- \alpha e), \gamma=\alpha \beta, 0 \leq \gamma e < 1/3 \}.
		\]
		 which yields a value of $1/8.15$.
		
\section{Conclusion and Open Problems}
\label{sec:conclusion}
		In this work, we described  two unifying ideas, namely non-uniform attenuation and multiple-chance probing to get bounds matching integrality gap (up to lower-order terms) for the \kcspip and its stochastic counterpart \sksp. We generalized the conjecture due to F\"uredi \etal~\cite{FKS93} (FKS conjecture) and went ``halfway'' toward resolving this generalized form using our ideas. Finally, we showed that we can improve the contention resolution schemes for \ufp with unit demands. Our algorithms for \kcspip can be extended to non-negative submodular objectives via the machinery developed in Chekuri \etal~\cite{CVZ14} and the improved contention resolution scheme for \ufp with unit demands leads to improved approximation ratio for submodular objectives via the same machinery.
		
		This work leaves a few open directions. The first concrete problem is to completely resolve the FKS conjecture and its generalization. We believe non-uniform attenuation and multiple-chances combined with the primal-dual techniques from \cite{FKS93} could give the machinery needed to achieve this. Other open directions are less well-formed. One is to obtain stronger \LP relaxations for the \kcspip and its stochastic counterpart \sksp such that the integrality gap is reduced. The other is to consider improvements to related packing programs, such as column-restricted packing programs or general packing programs.

\section*{Acknowledgements}
	 The authors would like to thank David Harris, the anonymous reviewers of SODA 2018 and the anonymous reviewers of TALG for useful suggestions which led to this improved version of the article. In particular, the authors would like to thank David Harris for pointing out the related works for the distributed implementation of graph coloring and showing the equivalence between the Conjecture~\ref{conj:fks} and Conjecture~\ref{conj:fks-ours}. Additionally, the authors are also grateful to Marek Adamcyzk who communicated his new (independent) result for the \sksp problem to us.

\bibliographystyle{ACM-Reference-Format}
\bibliography{refs}

\appendix
\section{Technical Lemmas}
\label{sec:techlemmas}
\kaedit{
In the main section, we use the following two variants of the Chernoff-Hoeffding bounds in the analysis of \kcspip algorithm. Theorem~\ref{appx:chernoffmult} is the standard multiplicative form while Theorem~\ref{appx:chernoff} can be derived from the standard form.

\begin{theorem}[Multiplicative form of Chernoff-Hoeffding bounds]
	\label{appx:chernoffmult}
		Let $X_1, X_2, \ldots, X_n \in [0, 1]$ be independent random variables such that $\mathbb{E}\left[ \sum_{i \in [n]} X_i \right] \leq \mu$. Then for every $\delta > 0$ we have,
		\[
			 \textstyle \Pr[\sum_{i \in [n]} X_i \geq (1+ \delta)\mu] \leq \left( \frac{e^{\delta}}{(1+\delta)^{(1+\delta)}} \right)^{\mu} \leq \exp\left[ -\tfrac{\delta^2 \mu}{2 + \delta} \right].
		\]
\end{theorem}
\begin{theorem}[Chernoff-Hoeffding bounds]
	\label{appx:chernoff}
	Suppose $c_1, c_2$, and $k$ are positive values with $c_2 \geq \frac{c_1}{k}$. Let $X_1, X_2, \ldots, X_n \in [0,1]$ be independent random variables satisfying $\mathbb{E}[\sum_{i}X_i] \leq \frac{c_1}{k}$. Then,
		\[
			\textstyle \Pr[\sum_{i}X_i \geq c_2] \leq \left( \frac{c_1 e}{k c_2} \right) ^{c_2}
		\]
\end{theorem}

\begin{proof}
	The standard form of the Chernoff-Hoeffding bounds (Theorem~\ref{appx:chernoffmult}) yields $\Pr[\sum_{i} X_i \geq (1+ \delta)\mu] \leq \left( \frac{e^{\delta}}{(1+\delta)^{(1+\delta)}} \right)^{\mu}$, for $\delta \geq 0$. 	
	Note that we want $(1+\delta)(c_1/k) = c_2$, hence giving $1+\delta = \frac{kc_2}{c_1}$. Plugging this into the standard form of the Chernoff-Hoeffding bounds gives us the desired bound.
\end{proof}
}

\begin{lemma}[Convexity]\label{lem:hm-convex}
	Assume $f:[0,1] \rightarrow [0,1]$ and it has second derivatives over $[0,1]$. Then we have that $\ln(1-t f(x))$ is a convex function of $x\in [0,1]$ for any given $t\in (0,1)$ iff $(1-f)(-f'') \ge f'^2$ for all $x \in [0,1]$. 
\end{lemma}
\begin{proof}
	Consider a given $t \in (0,1)$ and let $F(x)=\ln(1-t f(x))$. $F(x)$ is convex over $[0,1]$ iff $F'' \ge 0$ for all $x \in [0,1]$. We can verify that it is equivalent to the condition that $(1-f)(-f'') \ge f'^2$ for all $x \in [0,1]$.
\end{proof}
\kaedit{
\begin{lemma}
	\label{lem:HMtech}
	Let $0\leq a \leq 1$ be a given real number. Then for every $x \in [0, 1]$ we have that
	\[
			(1-a)^x \leq 1-ax.
	\]
\end{lemma}
\begin{proof}
	Consider the function $g(x) := 1-ax-(1-a)^x$. We will show that when $x \in [0, 1]$ we have $g(x) \geq 0$. This will complete the proof. Note that $g(0) = 0$. We will now show that $g(x)$ is increasing in $x \in [0, 1]$. We have that $g'(x) = (1-a) - (1-a)^x \ln (1-a) = (1-a) + (1-a)^x \ln \left( \tfrac{1}{1-a} \right) \geq 0$.	
\end{proof}

}
\section{Submodular Functions}
\label{appx:submodular}

In this section, we give the required background needed for submodular functions. 

\begin{definition}[Submodular functions]
A function $f: 2^{[n]} \rightarrow \bR_+^{+}$ on a ground-set of elements $[n] := \{1, 2, \ldots, n\}$ is called submodular if for every $A, B \subseteq [n]$, we have that $f(A \cup B) + f(A\cap B) \leq f(A) + f(B)$. Additionally, $f$ is said to be monotone if for every $A \subseteq B \subseteq [n]$, we have that $f(A) \leq f(B)$. 
\end{definition}

For our algorithms, we assume a \emph{value-oracle} access to a submodular function. This means that, there is an oracle which on querying a subset $T \subseteq [n]$, returns the value $f(T)$.

\begin{definition}[Multi-linear extension]
		The multi-linear extension of a submodular function $f$ is the continuous function $F: [0, 1]^n \rightarrow \bR_+^+$ defined as 
		\[
				\textstyle	F(x) := \sum_{T \subseteq [n]} (\prod_{j \in T} x_j \prod_{j \not \in T}(1-x_j)) f(T)
		\]
\end{definition}

Note that the multi-linear function $F(x) = f(x)$ for every $x \in \{0, 1\}^n$. The multi-linear extension is a useful tool in maximization of submodular objectives. In particular, the above has the following probabilistic interpretation. Let $S \subseteq [n]$ be a random subset of items where each item $i \in [n]$ is added into $S$ with probability $x_i$. We then have $F(x) = \bE_{S \sim x}[f(S)]$. It can be shown that the two definitions of $F(x)$ are equivalent. Hence, a lower bound on the value of $F(x)$ directly leads to a lower bound on the expected value of $f(S)$.

\end{document}